\documentclass[12pt]{amsart}

\usepackage{amssymb}
\usepackage{amsfonts}
\usepackage{amsthm}
\usepackage{amsmath,mathabx}
\usepackage{url}
\usepackage[colorlinks=true]{hyperref}
\usepackage{graphicx}
\usepackage{booktabs}
\usepackage{multirow}
\usepackage[a4paper, margin=2cm]{geometry}
\usepackage{siunitx}
\usepackage{cite}
\usepackage{setspace,kantlipsum} 
\usepackage{soul,color} 
\usepackage{ulem} 
\soulregister\cite7
\soulregister\ref7
\soulregister\pageref7
\soulregister{\item}{1}

\definecolor{airforceblue}{rgb}{0.36, 0.54, 0.66}
\definecolor{darkergreen}{rgb}{0.0, 0.5, 0.0}
\hypersetup{
    colorlinks,
    linkcolor={darkergreen},
    citecolor={airforceblue},
    urlcolor={blue}
}

\makeatletter

\newcommand{\vertj}[1]{{\vert\kern-0.25ex\vert\kern-0.25ex\vert #1
    \vert\kern-0.25ex\vert\kern-0.25ex\vert}}
\makeatother

\DeclareMathOperator{\Corr}{Corr}
\DeclareMathOperator{\Cov}{Cov}
\DeclareMathOperator{\Var}{Var}

\DeclareMathOperator{\E}{\mathbb E}

\theoremstyle{plain}
\newtheorem{Theorem}{Theorem}

\newtheorem{Definition}{Definition}

\newtheorem{Assumption}{Assumption}
\theoremstyle{definition}
\newtheorem{Example}{Example}
\newtheorem{Remark}{Remark}

\allowdisplaybreaks
\begin{document}

\title[Canonical correlation analysis]{Two approaches to multiple canonical correlation analysis for repeated measures data}

\author{Tomasz G\'{o}recki}
\address[T.~G\'{o}recki]{Faculty of Mathematics and Computer Science\\
  Adam Mickiewicz University\\
  Pozna\'n, Poland}
\email{tomasz.gorecki@amu.edu.pl}

\author{Miros\l aw Krzy\'sko}
\address[M.~Krzy\'sko]{Interfaculty Department of Mathematics and Statistics\\
University of Kalisz\\
Kalisz, Poland}
\email{m.krzysko@uniwersytetkaliski.edu.pl}

\author{Felix Gnettner}
\address[F.~Gnettner]{Department of Statistics\\
Colorado State University\\
Fort Collins, CO, USA}
\email{Felix.Gnettner@colostate.edu}

\author{Piotr Kokoszka}
\address[P.~Kokoszka]{Department of Statistics\\
Colorado State University\\
Fort Collins, CO, USA}
\email{Piotr.Kokoszka@colostate.edu}

\begin{abstract}
In classical canonical correlation analysis (CCA), the goal is to determine the linear transformations of two random vectors into two new random variables that are most strongly correlated. Canonical variables are pairs of these new random variables, while canonical correlations are correlations between these pairs. In this paper, we propose and study two generalizations of this classical method:

(1) Instead of two random vectors we study more complex
data structures that appear in important  applications.
In these structures, there are $L$ features, each described by
$p_l$ scalars, $1 \le l \le L$. We observe $n$ such objects over $T$ time points.
We derive a suitable analog of the CCA for such data. Our approach
relies on embeddings into Reproducing Kernel Hilbert Spaces, and covers
several related  data structures as well.

(2) We develop an analogous approach for multidimensional random processes.
In this case, the experimental units are  multivariate
continuous, square-integrable functions over a given interval.
These functions are modeled as elements of a Hilbert space,
so in this case, we define the  multiple
functional canonical correlation analysis, MFCCA.

We justify our approaches by their application to two
data sets and suitable large sample theory.
We derive consistency rates for the related transformation and correlation estimators, and show that it is possible to relax two common
assumptions on the compactness of the underlying cross-covariance
operators and the independence of the data.
\end{abstract}


\maketitle

\onehalfspacing

\section{Introduction}\label{sec:intro}
Canonical Correlation Analysis (CCA), proposed by \cite{Hotelling1936},
is often used to study relationships between two sets of features \cite{adrover2015, Krafty2013, Langworthy2021, Ma2020, Shu2020}. Linear transformations of the features in both sets (canonical variables) are constructed so that they are not correlated within each of the two  sets, but  the correlations between them  (canonical correlations) are maximal.
The objective is to maximize the correlation between data projections in the feature spaces.
This reveals the underlying structural relationships between these two feature sets, determining how much variability one set explains in the other. Such approaches are useful to determine if there exists a linear, or in more advanced cases nonlinear, mapping that transforms the sets of features into each other.

In this paper, the concepts and techniques of CCA are analyzed in  the case of more than two sets of features. Multiple CCA aims to  identify underlying patterns of correlation across more than two feature sets by finding transformations  of the variables within each set that are maximally correlated across the sets. Multiple CCA extends the concept of CCA to find shared relationships in complex, multi-set data, such as across multiple data matrices or different types of measurements, helping to reveal common structures or patterns that might be missed by analyzing pairs of sets individually. After estimating the optimal transformations, pairwise scatter plots of the transformed feature sets offer a powerful tool to visualize these patterns and structures.

\subsection*{A motivating data set}
To motivate the methodology developed in this paper, we consider the
Global Competitiveness  Index (GCI) data set. This example is developed further in  Section \ref{ss:gci}.
 For $n=115$ countries, $L=12$ features of economic competitiveness have been recorded. Two of them are institutions ($l=1$) and infrastructure ($l=2$). To each feature $l$ belongs a number $p_l$ of scalar indicators. The data have been recorded annually over $T=10$ years.

\begin{figure}[!tbh]
	\centering\includegraphics[width = 1\textwidth]{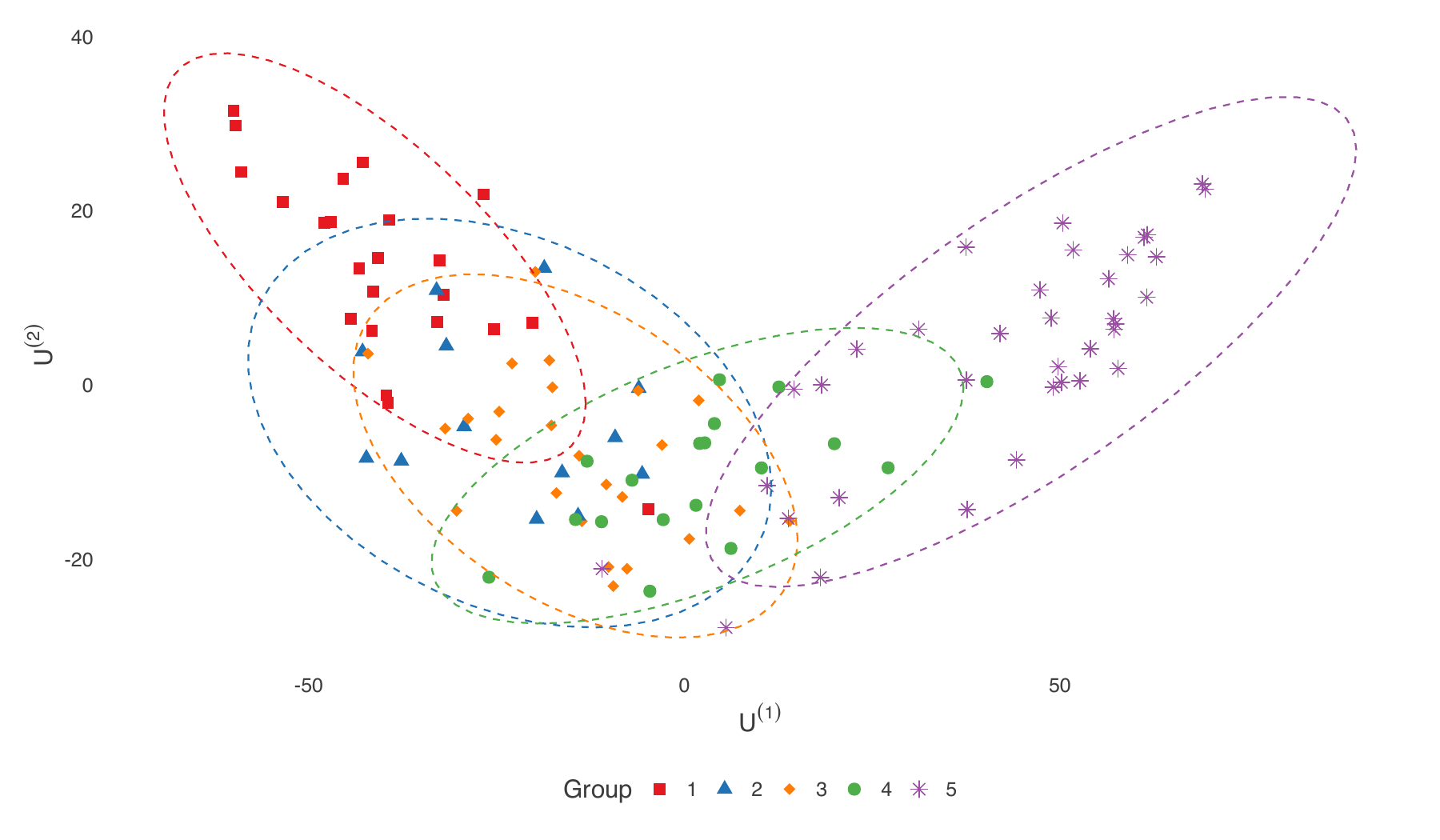}
	\caption{Scatterplot for the optimally transformed feature pairs in the GCI data set (115 countries and five groups)
in the system of the first two multiple kernel canonical variables $(U^{(1)}, U^{(2)})$ (with 95\% confidence normal ellipses). The optimal transformations were determined by multiple kernel CCA for multivariate repeated measures data, as described in Section~\ref{sec:kcca}.}
\label{fig:cca3}
\end{figure}

For a feature $l$ in country $k$, these observations can be encoded in a $T \times p_l$ matrix
\[
\pmb A_l[k] =
\begin{pmatrix}
a_{11} & \cdots & a_{1 p_l} \\
\vdots & \vdots & \vdots \\
a_{T1} & \cdots & a_{T p_l} \\
\end{pmatrix}, \ k \in \{1,\ldots,n\}, \quad l \in \{1,\ldots,L\}.
\]
The columns of the matrix $\pmb A_l[k]$ can be thought of as time courses, or repeated
measurements, of the $p_l$ indicators for country $k$.
The number and nature of scalar indicators may vary from feature to feature, but the
experimental units and number $T$ must be the same across all features.

The most elementary version of the multiple canonical correlations problem in the above setting is the optimization problem
\begin{align} \label{e:P1}
\max_{\pmb u_1, \ldots, \pmb u_L}\sum_{i=1}^L\sum_{j=1}^L
\Corr(\pmb A_i[1]\pmb u_i, \pmb A_j[1]\pmb u_j),
\end{align}
where for each feature $l \in \{1,\ldots,L\}$,  only $n=1$ country at only $T=1$ time point are considered.
This optimization problem is a population version; the correlations
are not estimated.
Each vector $\pmb u_l$ is called a ``weight vector'', and the vector
$\pmb y_l = \pmb A_l[1]\pmb u_l$ is called a ``component''.
The objective is to find weight vectors
\[[\pmb u_1^*,\ldots,\pmb u_L^*] = \arg\max_{\pmb u_1, \ldots, \pmb u_L}\sum_{i=1}^L\sum_{j=1}^L
\Corr(\pmb A_i[1]\pmb u_i, \pmb A_j[1]\pmb u_j)
\]
that lead to linear feature combinations within each block
with maximum correlation across all subjects.

For most applications, it is impossible to solve \eqref{e:P1}, because the underlying population correlation values are unknown. Thus, relaxing the problem in \eqref{e:P1} via replacing the correlation expressions by estimators of the related covariances, based on all $n=115$ countries, and introducing a suitable normalization constraint is required. For a concise mathematical expression, we refer to \eqref{e:sample_prob}.

The solution of such a sample-based problem is an estimator of the corresponding population optimum $[\pmb u_1^*,\ldots,\pmb u_L^*]$, which is required to compute the transformation mappings
\begin{equation}\label{e:transf}
    U^{(l)}(\cdot) = \langle {\pmb u}_l^*,\cdot \rangle,\qquad l \in\{1,\ldots,L\}.
\end{equation}
Plugging in each data point in this mapping generates a representation of the multiple CCA results. For a feature pair of interest, a scatter plot of the transformed data is produced, as depicted in Figure~\ref{fig:cca3} for the pair $(1,2)$. The different colors in this plot arise from the fact that the countries have been, independently from the analysis, divided into 5 groups by experts of the World Economic Forum. The multiple CCA was performed globally with respect to all countries, and it is easy to see that the clusters in Figure~\ref{fig:cca3} partially correspond to the group assignments.

\subsection*{Connections to previous research}
The multiple CCA problem was first considered by
Horst \cite{Horst1961a} whose solution is  called the
``maximum correlation method''. It, however, suffers from
severe problems with numerical convergence, \cite{berge1988}.
Kettenring \cite{Kettenring1971} proposed a different solution (also for $T = 1$)
to this problem, and named it the SUMCOR.
Also relevant to our  work, \cite{Carroll1968a,Hwang20123, Hwang2012} extended the CCA to the case of several data sets. Their method is known as the
``generalized canonical correlation analysis (GCCA)'' or the
``multiple-set canonical correlation analysis (MCCA)''.
It has been considered by many authors, e.g.,
\cite{Gower1989, Gloaguen2022, Lafosse1989, Park1996, Tenenhaus2011,
Tenenhaus2014a, Tenenhaus2015,Gifi1990, Yanai1998,Tenenhaus2017}.

Canonical correlation analysis for univariate functional data was introduced
in~\cite{Leurgans1993}, who studied two functional data sets.
It is also explained in Chapter 4 of \cite{HKbook} where references to
more recent research in the case of two functional data sets are given.
The CCA
of two data sets of Multidimensional Functional Data (MFCCA) was studied in
\cite{Gorecki2017, Gorecki2018, krzysko2019}.

\subsection*{Contributions}
We extend the methodology of multiple CCA to the new setting of
time-dependent observations of the multiple features of interest
obtained from different experimental units. We propose two approaches. 
Our first approach relies on kernel embeddings of each 
block into a Reproducing Kernel Hilbert space. 
We derive consistency rates that cover data scenarios in which the 
cross-covariance operators are not compact or the observations 
in the sample have some dependence structure among the experimental units. 
Neither case, to the best of our knowledge, has been studied, but they are 
relevant for data scenarios that motivate our
work. Assuming  independence and compactness is often an assumption 
of convenience, so our theory has practical relevance.  
Our second approach is an extension of multiple CCA to functional data 
indexed by time, which we think is particularly suitable for the
problems that motivate this work. In the functional context, we view the
repeated measurements as smooth time courses.  We show 
that the consistency rates derived for the multiple kernel CCA carry over 
to this extension. To illustrate the functional context, 
in the competitiveness index
example, suppose $Y_{l,j}^{(k)}(t)$
denotes the  value of indicator $j$ in the feature $l$ for country
$k$ in year $t$.  Unless a country experiences a dramatic change, like a
government overthrow or a financial crash, it is reasonable to assume that
the values of the $Y_{l,j}^{(k)}(t)$ evolve smoothly from year to year. 
This will be
true for other examples, like the performance of sectors of the economy 
or agriculture.
In such cases, a regularization of the time trajectories might lead to 
better-performing procedures.

Unlike in regression,
no specific response variables exist in the data structures we study. 
Our methodology
does not rely on assigning the roles of predictor and response variables.
The consistency of multiple CCA has not yet been theoretically analyzed
in the literature.

Our framework provides a unified view of several extensions of 
CCA, allowing both finite-dimensional and functional data to be treated 
within the same operator-theoretic setting. This perspective facilitates 
a rigorous comparison between kernel- and function-based approaches.  
The proposed approaches may serve as a foundation for 
advances  in consistent inference for high-dimensional dependent data
and regularized and sparse variants of multiple CCA.

\subsection*{Organization of the paper}
In Section~\ref{sec:kcca}, we derive the consistency rates of the multiple kernel CCA for repeated measures, while in Section~\ref{sec:fcca} multiple functional CCA is introduced.
Section~\ref{s:ex} contains real data studies that demonstrate the usefulness of and compare the proposed methods.
A summary and conclusions are presented in Section~\ref{sec:final}.

\section{Multiple kernel canonical correlation analysis in the case of repeated measures data}
\label{sec:kcca}
In this section,  we introduce multiple canonical correlation analysis and the corresponding estimators based on repeated measures data. We show that the consistency results in CCA can be extended to multiple CCA, including cases where

\begin{itemize}
    \item there is dependence in the data among the experimental units. This often arises when the units represent different time periods, or are related to each other because of their location.
    \item the underlying covariance operators are not compact. This is a common assumption in the literature, but it cannot be verified in practice. We showcase a simple example, where this assumption can be dropped.
\end{itemize}

Our theory builds on the classical kernel canonical correlation analysis problem \cite{akaho2001, alam2010, alam2013, alam2015, bach2002, bilenko2016,
Fukumizu2007, Hardoon2004, Hardoon2009, lai2000}, where only two features are considered, i.e.\ $L = 2$.
We begin the
with a review of relevant definitions related to Reproducing Kernel Hilbert Spaces
(RKHS) and their selected properties in order to facilitate the exposition that follows.

Suppose $[X_1,\ldots,X_L]$, are possibly dependent random elements with joint distribution\linebreak $P^{[X_1,\ldots,X_L]}$, taking values in a general space $\bigtimes_{l=1}^L \mathcal S_l$. The random elements $[X_1,\ldots,X_L]$ correspond to $[\pmb A_1[1],\ldots,\pmb A_L[1]]$ in our motivating data example in the population case of $n=1$, where no repeated measurements are available.

\begin{Definition}\label{def:RKHS}
For each $l \in \{1,\dots,L\}$ consider a measurable positive definite \textit{kernel} $K_l: \mathcal S_l \times \mathcal S_l \to \mathbb R$, and the associated kernel embedding $x \mapsto K_l(\cdot,x)$ of $x \in \mathcal S_l$ into the reproducing kernel Hilbert space $\mathcal H_{K_l}$. The Moore-Aronszajn theorem \cite{aronszajn1950} guarantees that the inner product in $\mathcal{H}_{K_l}$
determines the kernel, i.e.
	\[K_l(x,y) = \langle K_l(\cdot,x), K_l(\cdot,y)\rangle_{\mathcal{H}_{K_l}} \textrm{ for all } x,y\in \mathcal S_l.\]
\end{Definition}

Based on the notation introduced above we are now able to define the population version of the optimization problem for multiple kernel canonical correlation analysis.

\begin{Definition}
    The population multiple kernel CCA problem in the RKHS $\bigtimes_{l=1}^L \mathcal H_{K_l}$ is defined as
    \[ \max_{g_1 \in \mathcal H_{K_1},\ldots,g_L \in \mathcal H_{K_L}} \sum_{i=1}^L \sum_{j \in \{1,\ldots,L\}\setminus \{i\}} \Cov(g_i(X_i),g_j(X_j))
    \quad \text{s.t. } \sum_{k=1}^L\Var(g_k(X_k)) = L.
    \]
\end{Definition}

    In contrast to Section~\ref{sec:intro}, this optimization problem is formulated
    in the context of an RKHS rather than in the linear space spanned by the feature vectors. It thus allows us to deal with nonlinear dependencies between the features.
   For $i \neq j$, $\pmb C_{i,j}: \mathcal H_{K_j} \to \mathcal H_{K_i}$ denotes the cross-covariance operator between $X_i$ and $X_j$ and $\pmb C_{i,i}: \mathcal H_{K_i} \to \mathcal H_{K_i}$ is the covariance operator of $X_i$.

\begin{Remark}
    Using the above covariance operator notation, the multiple kernel CCA problem can be written as
    \begin{align}\label{e:optimization_problem}
        \rho_{\mathcal F}=\max_{g_1 \in \mathcal H_{K_1},\ldots,g_L \in \mathcal H_{K_L}} \sum_{i=1}^L \sum_{j \in \{1,\ldots,L\}\setminus \{i\}} \langle g_i,\pmb C_{i,j} \ g_j \rangle_{\mathcal H_{K_i}} \quad \text{s.t. } \sum_{k=1}^L \langle g_k,\pmb C_{k,k} \ g_k \rangle_{\mathcal H_{K_k}} = L.
    \end{align}
\end{Remark}
At first glance, it is not clear whether the optimization problem has a solution that is easy to determine. The following remark provides a clarification.
\begin{Remark}\label{r:eig_prob}
    The solution of the multiple kernel CCA problem is given by the appropriately normalized eigenfunctions belonging to the largest eigenvalue $\rho_{\mathcal F}$ of the following generalized eigenvalue problem:
\begin{align}\label{e:gen_eigenproblem}
    \sum_{j \in \{1,\ldots,L\}\setminus \{i\}} \pmb C_{i,j} \ g_j = \rho \cdot \pmb C_{i,i}\ g_i, \quad i \in \{1,\ldots,L\}.
\end{align}
This can be easily seen by computing the gradient of the Lagrangian of \eqref{e:optimization_problem}, i.e.
\[\nabla \mathcal L(g_1,\ldots,g_L;\lambda) =\begin{pmatrix}\sum_{j \in \{1,\ldots,L\}\setminus \{1\}} \pmb C_{1,j} \ g_j\\ \vdots \\ \sum_{j \in \{1,\ldots,L\}\setminus \{L\}} \pmb C_{L,j} \ g_j,\\ \sum_{k=1}^L \langle g_k,\pmb C_{k,k} \ g_k \rangle_{\mathcal H_{K_k}}- L \end{pmatrix} + \lambda \cdot \begin{pmatrix}\pmb C_{1,1}\ g_1\\ \vdots \\ \pmb C_{L,L}\ g_L \\ 0 \end{pmatrix}.\]
Solving $\nabla \mathcal L(g_1,\ldots,g_L;\lambda)= (0,\ldots,0)^T$, where the last component is only for normalization purposes, leads to the generalized eigenvalue problem \eqref{e:gen_eigenproblem}, which can be written as
\[\begin{pmatrix}
    0 & \pmb C_{1,2} & \pmb C_{1,3} & \cdots & \pmb C_{1,L}\\
    \pmb C_{2,1} & 0 & \pmb C_{2,3} & \cdots & \pmb C_{2,L}\\
    \pmb C_{3,1} & \pmb C_{3,2} & 0 & \cdots & \pmb C_{3,L}\\
    \vdots & \vdots & \vdots & \cdots & \vdots\\
    \pmb C_{L,1} & \pmb C_{L,2} & \pmb C_{L,3} & \cdots & 0
\end{pmatrix} \begin{pmatrix}
    g_1\\ g_2\\ g_3\\ \vdots \\ g_L
\end{pmatrix} = \rho
\begin{pmatrix}
    \pmb C_{1,1} & 0 & 0 & \cdots & 0\\
    0 & \pmb C_{2,2} & 0 & \cdots & 0\\
    0 & 0 & \pmb C_{3,3} & \cdots & 0\\
    \vdots & \vdots & \vdots & \cdots & \vdots\\
    0 & 0 & 0 & \cdots & \pmb C_{L,L}
\end{pmatrix}  \begin{pmatrix}
    g_1\\ g_2\\ g_3\\ \vdots \\ g_L
\end{pmatrix}. \]
If $\pmb C_{1,1},\ldots,\pmb C_{L,L}$ are strictly positive definite, this generalized eigenvalue problem can be easily transformed into an ordinary eigenvalue problem, which reads\begin{multline}\label{e:eigenvalue_problem}
    \begin{pmatrix}
    0 &\pmb C_{1,1}^{-1/2} \ \pmb C_{1,2} \ \pmb C_{2,2}^{-1/2} &\pmb C_{1,1}^{-1/2} \ \pmb C_{1,3} \ \pmb C_{3,3}^{-1/2} & \cdots &\pmb C_{1,1}^{-1/2} \ \pmb C_{1,L} \ \pmb C_{L,L}^{-1/2}\\
    \pmb C_{2,2}^{-1/2} \ \pmb C_{2,1} \ \pmb C_{1,1}^{-1/2} & 0 & \pmb C_{2,2}^{-1/2} \ \pmb C_{2,3} \ \pmb C_{3,3}^{-1/2} & \cdots & \pmb C_{2,2}^{-1/2} \ \pmb C_{2,L} \ \pmb C_{L,L}^{-1/2}\\
    \pmb C_{3,3}^{-1/2} \ \pmb C_{3,1} \ \pmb C_{1,1}^{-1/2} & \pmb C_{3,3}^{-1/2} \  \pmb C_{3,2} \ \pmb C_{2,2}^{-1/2} & 0 & \cdots &\pmb C_{3,3}^{-1/2} \  \pmb C_{3,L} \ \pmb C_{L,L}^{-1/2}\\
    \vdots & \vdots & \vdots & \cdots & \vdots\\
    \pmb C_{L,L}^{-1/2} \ \pmb C_{L,1} \ \pmb C_{1,1}^{-1/2} &\pmb C_{L,L}^{-1/2} \ \pmb C_{L,2} \ \pmb C_{2,2}^{-1/2} &\pmb C_{L,L}^{-1/2} \  \pmb C_{L,3} \ \pmb C_{3,3}^{-1/2} & \cdots & 0
\end{pmatrix}\cdot \begin{pmatrix}
    f_1\\ f_2\\ f_3\\ \vdots \\ f_L
\end{pmatrix}
 \\= \rho  \begin{pmatrix}
    f_1\\ f_2\\ f_3\\ \vdots \\ f_L
\end{pmatrix}.
\end{multline}
\end{Remark}

\begin{Remark}
    Considering the optimization problem
    \begin{equation}\label{e:problem_too_general}
        \max_{g_1 \in \mathcal H_{K_1},\ldots,g_L \in \mathcal H_{K_L}} \sum_{i=1}^L \sum_{j \in \{1,\ldots,L\}\setminus \{i\}} \langle g_i,\pmb C_{i,j} \ g_j \rangle_{\mathcal H_{K_i}} \quad \textrm{s.t. } \langle g_k,\pmb C_{k,k} \ g_k \rangle_{\mathcal H_{K_k}} = 1 \textrm{ for all } k \in \{1,\ldots,L\}
    \end{equation}
    instead of \eqref{e:optimization_problem} leads to mathematical difficulties. As pointed out in \cite[Section A, page 296]{Nielsen2002}, \eqref{e:problem_too_general} cannot be reduced to an ordinary generalized eigenvalue problem.
\end{Remark}

In the setting of multiple kernel CCA with repeated measurements, which is the sample version of the problem in \eqref{e:optimization_problem}, we deal with observations $[X_1^{(1)},\ldots,X_L^{(1)}],\ldots$, $[X_1^{(n)},\ldots,X_L^{(n)}]$ that can be assumed to be iid for a gentle introduction into the topic. The independence assumption can be replaced by some notion of weak dependence, see Example~\ref{example2}.

Replacing the population covariance operators in \eqref{e:optimization_problem} by corresponding estimators
\begin{multline}\label{e:cov_estimator}
    \langle g_i, \pmb{\widehat{C}}_{i,j}^{(n)}  \ g_j \rangle_{\mathcal H_{K_i}} = \frac 1n \sum_{k=1}^n \left\langle g_i,K_i(\cdot,X_i^{(k)})-\frac 1n \sum_{\ell=1}^n K_i(\cdot,X_i^{(\ell)}) \right\rangle_{\mathcal H_{K_i}}\\ \cdot \left\langle g_j,K_j(\cdot,X_j^{(k)})-\frac 1n \sum_{\ell=1}^n K_j(\cdot,X_j^{(\ell)}) \right\rangle_{\mathcal H_{K_j}}
\end{multline}
yields
\begin{equation}\label{e:sample_prob}
    \widehat \rho_{\mathcal F} = \max_{g_1 \in \mathcal H_{K_1},\ldots,g_L \in \mathcal H_{K_L}} \sum_{i=1}^L \sum_{j \in \{1,\ldots,L\}\setminus \{i\}} \langle g_i,\pmb{\widehat{C}}_{i,j}^{(n)} \ g_j \rangle_{\mathcal H_{K_i}} \quad \text{s.t. } \sum_{k=1}^L \langle g_k,\pmb{\widehat{C}}_{k,k}^{(n)} \ g_k \rangle_{\mathcal H_{K_k}} = L.
\end{equation}

In contrast to the previous population-based problems, \eqref{e:sample_prob} incorporates covariance estimators based on data from a sample of $n$ experimental units.

Solving this optimization problem via finding the solutions of a generalized eigenvalue problem as in Remark~\ref{r:eig_prob} requires some modification to keep it numerically tractable. The inversion of the covariance operator estimates $\pmb{\widehat C}_{i,i}^{(n)}$ is required for this purpose, but these usually do not have full rank. Consequently, some form of regularization is needed.
Various regularizations have been studied in similar contexts, \cite{alam2010,alam2013,bilenko2016,Tenenhaus2011,
Tenenhaus2014a,Tenenhaus2015,Tenenhaus2017}.
Basically, without any regularization, it is possible to find an estimate $\widehat \rho_{\mathcal F}$ that indicates perfect dependence, and does not entail any meaningful information on the correlations of the features.

Replacing $\pmb{\widehat C}_{i,i}^{(n)}$ by $\pmb{\widehat C}_{i,i}^{(n)} + \epsilon_n \cdot \pmb I$ with a positive decreasing sequence $(\epsilon_n)_{n \in \mathbb N}$, $\epsilon_n \to 0$, solves the invertibility issues. So, the regularized problem reads
\begin{equation}\label{e:reg_eig_prob}
    \widehat \rho_{\mathcal F} = \max_{g_1 \in \mathcal H_{K_1},\ldots,g_L \in \mathcal H_{K_L}} \sum_{i=1}^L \sum_{j \in \{1,\ldots,L\}\setminus \{i\}} \langle g_i,\pmb{\widehat{C}}_{i,j}^{(n)} \ g_j \rangle_{\mathcal H_{K_i}} \quad \text{s.t. } \sum_{k=1}^L \langle g_k,(\pmb{\widehat C}_{i,i}^{(n)} + \epsilon_n \cdot \pmb I)\ g_k \rangle_{\mathcal H_{K_k}} = L.
\end{equation}

For the sake of lucidity, in the
following, we state the eigenvalue problem \eqref{e:eigenvalue_problem} as
\begin{equation}\label{eq:CovOpMat}
    \mathfrak C_{L} \ \mathfrak f = \rho \ \mathfrak f,
\end{equation}
and the regularized sample version \eqref{e:reg_eig_prob} reads
\begin{align}\label{e:CCAproblem_short}
    \widehat{\mathfrak C}_{n,L} \ \widehat{\mathfrak f}_n = \widehat \rho \ \widehat{\mathfrak f}_n.
\end{align}

To obtain a simple numerical solution of \eqref{e:reg_eig_prob}, the Gram matrices of each feature are considered to get a computationally tractable formulation of the generalized eigenvalue problem.
\begin{Definition}
Given a  kernel $K_l$ and a sample
$X_l^{(1)},\ldots,X_l^{(n)}$ with values in $\mathcal S_l$,
the $n\times n$ matrix $\pmb G_l$ with entries $K_l(X_l^{(i)}, X_l^{(n)})$
is called  the \textit{Gram matrix}, or the \textit{kernel matrix},
of $K_l$ with respect to the given sample. $\pmb G_l$ is \textit{positive semi-definite} if $\pmb{{c}}^\top\pmb G_l\pmb{{c}}\geq 0 $ holds for any $\pmb c\in\mathbb{R}^n$
\end{Definition}

The centered kernel matrices $\pmb{\widetilde G}_l$ are defined by~\cite{Scholkopf2008} as
\begin{equation} \label{e:G-widetilde}
\pmb{\widetilde G}_l = \pmb H\pmb G_l\pmb H, \ \ \ \pmb H = \pmb I_n-\frac{1}{n}\pmb1_n\pmb1_n^\top, \quad l \in \{1,\ldots,n\}.
\end{equation}
The matrix $\pmb H$ is called
the \textit{centering matrix}; $\pmb I_n$ is the identity matrix of order $n$ and $\pmb 1_n$ is the column-vector with $n$ $1$s. Elementary linear algebra shows that the entries of $\pmb{\widetilde G}_{l}$ can be written as
\begin{multline}\label{e:gram_matrix}
        (\pmb{\widetilde G}_{l})_{k,\ell} = K_l(X_l^{(k)},X_l^{(\ell)}) - \frac 1n \sum_{b=1}^n K_l(X_l^{(k)},X_l^{(b)}) - \frac 1n \sum_{a=1}^n K_l(X_l^{(a)},X_l^{(\ell)})\\ + \frac{1}{n^2} \sum_{a=1}^n\sum_{b=1}^n K_l(X_l^{(a)},X_l^{(b)}).
\end{multline}

In our motivating data example, $\mathcal S_l = \mathbb{R}^{T\times p_l}$  and the repeated measurements $X_l^{(1)},\ldots,X_l^{(n)}$ correspond to the blocks $\pmb A_l[1],\ldots,\pmb A_l[n]$ for each feature $l \in \{1,\ldots,n\}$. In this context, we consider Gaussian kernels $K_l: \mathbb{R}^{T\times p_l} \times \mathbb{R}^{T\times p_l} \to \mathbb R$, so that the corresponding Gram matrix has the entries
\begin{equation}\label{e:GK}
    (\pmb G_l)_{i,j} = K_l(\pmb A_l[i],\pmb A_l[j]) = \exp \left(-\gamma \|\pmb A_l[i]-\pmb A_l[j]\|_F^2 \right),\quad i,j \in\{1,\ldots, n\},
\end{equation}
where $\|\cdot\|_F$ denotes the Frobenius norm and $\gamma$ is a fixed positive constant.
Here, the kernel $K_l$ can be regarded as a similarity measure between two elements
$\pmb A_l[i]$ and $\pmb A_l[j]$, relevant to our specific real data analysis problem.

With the above preparation, we can formulate the regularized problem \eqref{e:reg_eig_prob} by means of the centered kernel Gram matrices in \eqref{e:gram_matrix}. Since the range of of each $\widehat{\pmb C}_{i,i}, i \in \{1,\ldots,L\},$ is spanned by $u_i^{(1)},\ldots,u_i^{(n)}$ with $u_i^{(k)} = K_i(\cdot,X_i^{(k)})-\frac 1n \sum_{\ell=1}^n K_i(\cdot,X_i^{(\ell)})$, we can write each sololution $g_1,\ldots,g_l$ of the MCCA problem as a linear combination, i.e. there exist vectors $\pmb w_1,\ldots,\pmb w_L \in \mathbb R^n$ such that \[g_i = \sum_{k=1}^n w_i^{(k)} \cdot u_i^{(k)}, \quad i \in \{1,\ldots,L\}.\] Thus, it holds
\begin{align*}
        \langle g_i, \pmb{\widehat{C}}_{i,j}^{(n)}  \ g_j \rangle_{\mathcal H_{K_i}} &= \frac 1n \sum_{k=1}^n \left\langle g_i,u_{i}^{(k)} \right\rangle_{\mathcal H_{K_i}} \cdot \left\langle g_j,u_{j}^{(k)} \right\rangle_{\mathcal H_{K_j}}\\
        &= \frac 1n \sum_{k=1}^n \left\langle \sum_{l_1=1}^n w_{i}^{(l_1)} \cdot u_{i}^{(l_1)},u_{i}^{(k)} \right\rangle_{\mathcal H_{K_i}} \cdot \left\langle \sum_{l_2=1}^n w_{j}^{(l_2)} \cdot u_{j}^{(l_2)},u_{j}^{(k)} \right\rangle_{\mathcal H_{K_j}}\\
        &= \frac 1n \sum_{k=1}^n \sum_{l_1=1}^n \sum_{l_2=1}^n w_{i}^{(l_1)} \left\langle u_{i}^{(l_1)},u_{i}^{(k)} \right\rangle_{\mathcal H_{K_i}}\left\langle u_{j}^{(l_2)},u_{j}^{(k)} \right\rangle_{\mathcal H_{K_j}}w_{j}^{(l_2)}= \frac{\pmb w_{i}^T\pmb{\widetilde G}_{i}\pmb{\widetilde G}_{j} \pmb w_{j}}{n},
\end{align*}
as well as
\[\langle g_i,(\widehat{\pmb C}_{i,i} + \epsilon_n \cdot \pmb I) \ g_i \rangle_{\mathcal H_{K_i}} = \frac{\pmb w_{i}^T\pmb{\widetilde G}_{i}^2 \pmb w_{i}}{n} + \epsilon_n \cdot  \pmb w_{i}^T\pmb{\widetilde G}_{i} \pmb w_{i}= \pmb w_{i}^T\left(\frac{\pmb{\widetilde G}_{i}^2}n + \epsilon_n \cdot \pmb{\widetilde G}_{i}\right)\pmb w_{i},\]
where $\pmb w_{i}^T = (w_{i}^{(1)},\ldots,w_{i}^{(n)})$. We made use of the identity $\langle K_i(\cdot,X_i^{(l_1)}),K_i(\cdot,X_i^{(l_2)}) \rangle_{\mathcal H_{K_i}} = K_i(X_i^{(l_1)},X_j^{(l_2)})$, which implies that $\langle u_{i}^{(k)},u_{i}^{(\ell)} \rangle_{\mathcal H_{K_i}} = (\pmb{\widetilde G}_{i})_{k,\ell}$.

These transformations yield the following formulation of the sample kernel MCCA problem \eqref{e:reg_eig_prob}:
\begin{equation}\label{e:reg_eig_prob_gram}
    \widehat \rho_{\mathcal F} = \max_{\pmb w_1,\ldots,\pmb w_L \in \mathbb R^n} \sum_{i=1}^L \sum_{j \in \{1,\ldots,L\}\setminus \{i\}}  \frac{\pmb w_{i}^T\pmb{\widetilde{G}}_{i}\pmb{\widetilde{G}}_{j} \pmb w_{j}}{n} \quad \text{s.t. } \sum_{k=1}^L \pmb w_{i}^T\left(\frac{\pmb{\widetilde{G}}_{i}^2}n + \epsilon_n \cdot \pmb{\widetilde{G}}_{i}\right) \pmb w_{i} = L.
\end{equation}
In block matrix notation, this reads
\begin{multline*}
\frac 1n \cdot \begin{pmatrix}
    0 & \pmb{\widetilde{G}}_1 \pmb{\widetilde{G}}_2 & \pmb{\widetilde{G}}_1 \pmb{\widetilde{G}}_3 & \cdots & \pmb{\widetilde{G}}_1 \pmb{\widetilde{G}}_L\\
    \pmb{\widetilde{G}}_2 \pmb{\widetilde{G}}_1 & 0 & \pmb{\widetilde{G}}_2 \pmb{\widetilde{G}}_3 & \cdots & \pmb{\widetilde{G}}_2 \pmb{\widetilde{G}}_L\\
    \pmb{\widetilde{G}}_3 \pmb{\widetilde{G}}_1 & \pmb{\widetilde{G}}_3 \pmb{\widetilde{G}}_2 & 0 & \cdots & \pmb{\widetilde{G}}_3 \pmb{\widetilde{G}}_L\\
    \vdots & \vdots & \vdots & \cdots & \vdots\\
    \pmb{\widetilde{G}}_L \pmb{\widetilde{G}}_1 & \pmb{\widetilde{G}}_L \pmb{\widetilde{G}}_2 & \pmb{\widetilde{G}}_L \pmb{\widetilde{G}}_3 & \cdots & 0
\end{pmatrix} \begin{pmatrix}
    \pmb w_1\\ \pmb w_2\\ \pmb w_3\\ \vdots \\ \pmb w_L
\end{pmatrix}\\ = \rho
\begin{pmatrix}
    \frac{\pmb{\widetilde{G}}_1^2}{n} + \epsilon_n \cdot \pmb{\widetilde{G}}_1 & 0 & 0 & \cdots & 0\\
    0 & \frac{\pmb{\widetilde{G}}_2^2}{n} + \epsilon_n \cdot \pmb{\widetilde{G}}_2 & 0 & \cdots & 0\\
    0 & 0 & \frac{\pmb{\widetilde{G}}_3^2}{n} + \epsilon_n \cdot \pmb{\widetilde{G}}_3 & \cdots & 0\\
    \vdots & \vdots & \vdots & \cdots & \vdots\\
    0 & 0 & 0 & \cdots & \frac{\pmb{\widetilde{G}}_\mathrm L^2}{n} + \epsilon_n \cdot \pmb{\widetilde{G}}_L
\end{pmatrix} \begin{pmatrix}
    \pmb w_1\\ \pmb w_2\\ \pmb w_3\\ \vdots \\ \pmb w_L
\end{pmatrix}.
\end{multline*}
Instead, the function space related representation \eqref{e:reg_eig_prob}, the above finite-dimensional representation is suitable for numerical solvers.
Now, we formulate the assumptions for our multiple kernel canonical correlations consistency result. Our theory is based on the proof strategies of \cite{Fukumizu2007}, but we relax several of their assumptions.

\begin{Assumption}\label{asm1}
    $\E(K_i(X_i,X_i))<\infty$ holds for all $i \in \{1,\ldots,L\}$.
\end{Assumption}
As we consider covariance operators in a RKHS setting, Assumption
\ref{asm1} is required to guarantee that the random variables $g_i(X_i)$ have finite second moments for $g_i \in \mathcal H_{K_i}, i \in \{1,\ldots, L\}$.

\begin{Assumption}\label{asm2}
    The covariance operators $\pmb C_{i,i}$, $i \in \{1,\ldots,L\}$, are strictly positive definite.
\end{Assumption}
Assumption \ref{asm2} ensures that the generalized eigenvalue problem can be transformed into the ordinary eigenvalue problem.

\begin{Assumption}\label{asm3}
    For the sequence $(\epsilon_n)_{n \in \mathbb N}$ it holds $\epsilon_n \to 0$ as well as $n^{1/3}\epsilon_n\to \infty$, as $n\to \infty$.
\end{Assumption}
As explained previously, some regularization is required because the covariance operator estimates do not have full rank and thus are not invertible. For our asymptotic considerations, the regularization sequence
must converge to 0, but not too fast.

\begin{Assumption}\label{asm4}
    The data $[X_1^{(1)},\ldots,X_L^{(1)}],\ldots,[X_1^{(n)},\ldots,X_L^{(n)}]$ have a dependence structure such that for any $(i,j) \in \{1,\ldots,n\}$ the estimator $\pmb{\widehat C}_{i,j}$ is weakly $\sqrt n$-consistent in operator norm, i.e. $O_P(1/\sqrt n)$,
\end{Assumption}
By the multivariate CLT, Assumption \ref{asm4} holds for under the usual
assumption of independence and identical distribution across $n$,
but also under notions of weak dependence that imply weak invariance principles.

While Assumptions \ref{asm1}-\ref{asm3} are fairly standard, Assumption~\ref{asm4} is more delicate. Below, we present two examples in which Assumption~\ref{asm4} is satisfied. Example~\ref{example1} highlights that it is possible to drop the requirement of compact cross-covariance operators, which is imposed throughout the literature. Example~\ref{example2} is focused on incorporating dependence between the experimental units. In fact, this is more realistic than assuming independence for most real data settings.

\begin{Example}\label{example1}
    Consider a bounded sample $[Z_1^{(1)},Z_2^{(1)}],\ldots,[Z_1^{(n)},Z_2^{(n)}]$ of iid pairs of random variables. $Z_1^{(i)}$ and $Z_2^{(i)}$ take values in $\mathrm L^2$, and $\|Z_1^{(i)}\|_{\mathrm L^2} \leq a_1$, $\|Z_1^{(i)}\|_{\mathrm L^2} \leq a_2$ holds for some fixed constants $a_1,a_2 >0$. Then, the covariance estimator
    \[\pmb{\widehat C}_{1,2} = \left(\frac 1n \sum_{i=1}^{n}Z_1^{(i)}{Z_2^{(i)}}^{\top}\right) - \left(\frac 1n \sum_{j=1}^n Z_1^{(j)}\right) \left(\frac 1n \sum_{\ell=1}^n{Z_2^{(\ell)}}\right)^{\top}\]
    is weakly $\sqrt n$-consistent. This can be seen by
\begin{align*}
    \|&\pmb{\widehat C}_{1,2} - \pmb{C}_{1,2}\|_{\mathrm{op}}\\
    &\leq \left\| \frac 1n \sum_{i=1}^{n}Z_1^{(i)}{Z_2^{(i)}}^{\top} - \mathbb E \left( Z_1^{(1)}{Z_2^{(1)}}^{\top}\right) \right\|_{\mathrm{op}} + \left\|\left(\frac 1n \sum_{j=1}^n Z_1^{(j)}\right) \left(\frac 1n \sum_{\ell=1}^n{Z_2^{(\ell)}} - \mathbb E \left( Z_2^{(1)}\right) \right)^{\top} \right\|_{\mathrm{op}}\\
    &\quad + \left\|\left(\frac 1n \sum_{j=1}^n Z_1^{(j)} -\mathbb E \left( Z_1^{(1)}\right) \right) \left( \mathbb E \left( Z_2^{(1)}\right) \right)^{\top} \right\|_{\mathrm{op}}\\
    &\leq \left\| \frac 1n \sum_{i=1}^{n}Z_1^{(i)}{Z_2^{(i)}}^{\top} - \mathbb E \left( Z_1^{(1)}{Z_2^{(1)}}^{\top}\right) \right\|_{\mathrm{op}} + a_1 \cdot \left\|_{\mathrm{op}} \left(\frac 1n \sum_{\ell=1}^n{Z_2^{(\ell)}} - \mathbb E \left( Z_2^{(1)}\right) \right)^{\top} \right\|_{\mathrm L^2}\\
    &\quad + a_2 \cdot \left\| \left(\frac 1n \sum_{j=1}^n{Z_1^{(j)}} - \mathbb E \left( Z_1^{(1)}\right) \right)^{\top} \right\|_{\mathrm L^2},
\end{align*}
    because the spectral norm of a dyadic product equals the product of the $\mathrm L^2$-norms of the vectors.
    For the first term on the right hand side, applying McDiarmid's inequality, respectively its generalization to random elements in Banach spaces \cite{Pinelis94}, implies weak $\sqrt n$-consistency, as the mean satisfies the bounded differences property (in the spectral norm) with upper bound
    \[\frac 1n\|Z_1^{(i)}{Z_2^{(i)}}^{\top}-\widetilde Z_1^{(i)}{{\widetilde Z}_2^{(i)^{\top}}}\|_{\mathrm{op}} \leq \frac{2a_1a_2}{n}.\] This is because the spectral norm of a dyadic product equals the product of the $\mathrm L^2$-norms of the vectors. Thus, for any $\epsilon>0$ McDiarmid's inequality entails
    \[\mathbb P\left(\left\| \frac 1n \sum_{i=1}^{n}Z_1^{(i)}{Z_2^{(i)}}^{\top} - \mathbb E \left( Z_1^{(1)}{Z_2^{(1)}}^{\top}\right) \right\|_{\mathrm{op}}>\epsilon\right) \leq \exp \left(-\frac{n\epsilon^2}{2a_1^2a_2^2} \right),\] and the layer cake representation of the expected value then gives that
    \[\mathbb E \left(\left\| \frac 1n \sum_{i=1}^{n}Z_1^{(i)}{Z_2^{(i)}}^{\top} - \mathbb E \left( Z_1^{(1)}{Z_2^{(1)}}^{\top}\right) \right\|_{\mathrm{op}}>\epsilon\right)= O\left(\frac{1}{\sqrt n}\right).\]
    Applying a similar argument with respect to the $\mathrm L^2$-norm to the other summands entails the same rate for them. Hence, $\mathbb E(\|\pmb{\widehat C}_{1,2} - \pmb{C}_{1,2}\|_{\mathrm{op}}) = O(1/\sqrt n)$.
    Thus, it is possible to not assume compactness of the cross-covariance operators $\pmb C_{i,j}$ if the sample, respectively, the embeddings of the sample are bounded. Non-compact cross-covariance operators can actually arise from very simple dependence structures, e.g.\ $Z_1^{(i)} = Z_2^{(i)}$, $i \in \{1,\ldots,N\}$. The latter is a very trivial example in the context of CCA, but useful to point out that there exist more complex dependence structures that cannot be described by compact cross-covariance operators. This was observed in \cite{Fukumizu2007}, who provided a sufficient condition for a cross-covariance operator to be Hilbert-Schmidt. However, this condition is not satisfied in general.
\end{Example}

\begin{Example}\label{example2}
    Consider a bounded sequence $[Z_1^{(1)},Z_2^{(1)}],\ldots,[Z_1^{(n)},Z_2^{(n)}]$ of identically distributed $\mathrm L^4$-$m$-approximable pairs of random variables, whose cross-covariance operator is Hilbert-\linebreak Schmidt. Then, Assumption~\ref{asm4} holds, as shown in Theorem~3 of \cite{KR2013}. The reasoning in the latter is for ordinary covariance operators in Hilbert spaces, but carries over to cross-covariance operators.
\end{Example}

Based on the above assumptions, we can now state the main result of this section, which provides consistency rates for the outputs of multiple kernel CCA. In particular, it  entails a joint consistency rate for the optimal transformations $[\pmb u_1^*,\ldots,\pmb u_L^*]$ that are required in \eqref{e:transf} to create plots as in Figure~\ref{fig:cca3}. The operator $\mathfrak C_{L}$ belongs to the population version of the CCA-related eigenvalue problem \eqref{eq:CovOpMat}, see also \eqref{e:gen_eigenproblem} for its detailed representation.
\begin{Theorem}\label{t:main}
        Let Assumptions \ref{asm1}-\ref{asm4} hold. Furthermore, we assume that $\mathfrak C_{L}$ is closed, bounded and has isolated eigenvalues of multiplicity one. Let $\widehat \rho_n$ be the estimator, arising from \eqref{e:CCAproblem_short} for an eigenvalue $\rho$ of $\mathfrak C_{L}$ and $\widehat{\mathfrak{f}}$  the corresponding normalized eigenvector estimator for $\mathfrak f$. Then,
        \begin{align*}
            |\widehat \rho_n - \rho| &= O_P(\max\{\epsilon_n,\epsilon_n^{-\frac 32} \cdot n^{-\frac 12}\}),\\
            \|\widehat{\mathfrak f}_n - \mathfrak f \|&=  O_P(\max\{\epsilon_n,\epsilon_n^{-\frac 32} \cdot n^{-\frac 12}\}).
        \end{align*}
\end{Theorem}
\begin{proof}
    If a sequence of bounded linear operators $(T_n)_{n \in \mathbb N}$ converges in the operator norm $\|\cdot\|_{\textrm{op}}$ to some bounded limit element $T$, i.e. $\|T_n-T\|_{\textrm{op}}\to 0$, this implies strong stability on the resolvent of $T$ \cite[Proposition 2.11]{Chatelin1981}. As all the eigenvalues of our considered covariance operators are isolated by assumption, Proposition 4.1 from \cite{Chatelin1981} entails that there exists a projection operator $\mathcal P$ such that \[|\widehat \rho_n - \rho| = O_P( \|(\widehat{\mathfrak C}_{n,L} - \mathfrak C_L) \mathcal P\|_{\textrm{op}}), \quad \|\widehat{\mathfrak f}_n - \mathfrak f \| =O_P( \|(\widehat{\mathfrak C}_{n,L} - \mathfrak C_L) \mathcal P\|_{\textrm{op}}).\]

    Since
    \begin{align*}
        \|(\widehat{\mathfrak C}_{n,L} - \mathfrak C_L) \mathcal P\|_{\textrm{op}} &\leq  \|\widehat{\mathfrak C}_{n,L} - \mathfrak C_L \|_{\textrm{op}}\\  &\leq \sum_{i,j=1,i\neq j}^L \|(\widehat{\pmb C}_{i,i} + \epsilon_n \cdot \pmb I)^{-1/2} \ \widehat{\pmb C}_{i,j} \ (\widehat{\pmb C}_{j,j} + \epsilon_n \cdot \pmb I)^{-1/2}-\pmb C_{i,i}^{-1/2} \ \pmb C_{i,j} \ \pmb C_{j,j}^{-1/2}\|_{\textrm{op}},
    \end{align*}
    it suffices to derive the consistency rate for the right hand side in order to obtain consistency rates for the eigenvector and eigenvalue estimators. Under Assumption~\ref{asm4}, Lemma 6 of \cite{Fukumizu2007} also holds for non-iid data $[X_1^{(1)},\ldots,X_L^{(1)}],\ldots,[X_1^{(n)},\ldots,X_L^{(n)}]$. Thus, whe have for $i \neq j$:
    \begin{multline}\label{e:lemma6fukumizu}
        \|(\widehat{\pmb C}_{i,i}+\epsilon_n\cdot \pmb I)^{-\frac 12} \ \widehat{\pmb C}_{i,j}  (\widehat{\pmb C}_{j,j}+\epsilon_n\cdot \pmb I)^{-\frac 12} - (\pmb C_{i,i}+\epsilon_n\cdot \pmb I)^{-\frac 12} \ \pmb C_{i,j}  (\pmb C_{j,j}+\epsilon_n\cdot \pmb I)^{-\frac 12}\|_{\textrm{op}}\\ = O_P \left(\epsilon_n^{-\frac 32} \cdot n^{-\frac 12}\right).
    \end{multline}
    Moreover, the upper bound
    \begin{align}\label{e:th_covop}
        &\|(\pmb C_{i,i}+\epsilon_n\cdot \pmb I)^{-\frac 12} \ \pmb C_{i,j}  (\pmb C_{j,j}+\epsilon_n\cdot \pmb I)^{-\frac 12} - \pmb C_{i,i}^{-\frac 12} \ \pmb C_{i,j}  \pmb C_{j,j}^{-\frac 12}\|_{\textrm{op}}\\
        &\leq \| ((\pmb C_{i,i}+\epsilon_n\cdot \pmb I)^{-\frac 12} - \pmb C_{i,i}^{-\frac 12}) \ \pmb C_{i,j} \ (\pmb C_{j,j}+\epsilon_n\cdot \pmb I)^{-\frac 12} \|_{\textrm{op}} +  \| \pmb C_{i,i}^{-\frac 12} \ \pmb C_{i,j} \ ((\pmb C_{j,j}+\epsilon_n\cdot \pmb I)^{-\frac 12} - \pmb C_{j,j}^{-\frac 12})\|_{\textrm{op}}\nonumber\\
        &= O(\max\{\|(\pmb C_{j,j}+\epsilon_n\cdot \pmb I)^{-\frac 12} - \pmb C_{i,i}^{-\frac 12}\|_{\textrm{op}},\|(\pmb C_{j,j}+\epsilon_n\cdot \pmb I)^{-\frac 12} - \pmb C_{j,j}^{-\frac 12}\|_{\textrm{op}} \}). \nonumber
    \end{align}
    holds, because the covariance operators are bounded and invertible, which implies boundedness of their inverse square root. By setting $\pmb A = \pmb C_{i,i}+\epsilon_n \cdot \pmb I$ and $\pmb B = \pmb C_{i,i}$ in
    \[\|\pmb A^{-\frac 12} - \pmb B^{-\frac 12}\|_{\textrm{op}} = \|\pmb A^{-\frac 12} \ (\pmb B^{\frac 32} - \pmb A^{\frac 32} ) \ \pmb B^{-\frac 32} + (\pmb A - \pmb B) \ \pmb B^{-\frac 32} \|_{\textrm{op}},\] and combining it with the fact from \cite[Lemma 8]{Fukumizu2007} that there exists a constant $\lambda>0$ for the choice of $\pmb A$ and $\pmb B$ such that \[\|\pmb A^{\frac 32} - \pmb B^{\frac 32}\|_{\textrm{op}} \leq 3\lambda^\frac{1}{2}\ \|\pmb A - \pmb B\|_{\textrm{op}},\] we obtain that the term on the right hand side of \eqref{e:th_covop} is $O(\epsilon_n)$. This concludes the proof.
\end{proof}

\begin{Remark}
    While \cite{Fukumizu2007} prove consistency of the estimators $\widehat \rho_n$ and $\widehat{\mathfrak f}_n$, but do not provide rates for them, \cite{Zhou20} give rates that directly rely on the decay of the eigenvalues of the underlying covariance operators, which is usually unknown. Their assumption on the decay requires compactness of the underlying covariance operators, which is not assumed in this paper. The results in \cite{Hsing15} do not cover regularized estimators, which are required in the sample case.
\end{Remark}

\begin{Remark}
    The number of features $L$ can also be allowed to slowly grow with the sample size $n$. In this case, an additional factor $L^2$ appears in the consistency rates. Nevertheless, this only leads to a convergence statement of the form $\|\widehat{\mathfrak C}_{n,L} - \mathfrak C_L\|\overset{P}{\to} 0 $ in which the operator norm also depends on $L$. Without defining a proper limit element of $\mathfrak C_L$, the interpretation of such a mathematical result is unclear.
\end{Remark}

\section{Multiple functional CCA} \label{sec:fcca}
Instead of embedding the matrix-valued observations $\pmb A_l[1],\ldots,\pmb A_l[n]$, $l \in \{1,\ldots,L \}$,  from our motivating example in Section~\ref{sec:intro} into an RKHS, we consider a different approach in this section. We transform each data block $\pmb A_l[i]$ into a function $\pmb Y_l[i]$, i.e.\ we are smoothing every block in order to lift it into a function space, before performing CCA with these transformed data.

Many smoothing procedures are well-known, see e.g.\ \cite{Gorecki2018, Ramsay2005}. We assume our smoothed observations to be of the form $\pmb Y_l[i]=(Y_{l,1}^{(i)},\ldots ,Y_{l,p_l}^{(i)})^{\top}$, where
each $Y_{l,j}^{(i)}$ is a function in $\mathrm \mathrm L^2(I)$, and $I$ is an interval. For the sake of lucidity, we only use the notation $\pmb Y_l[i](t)=(Y_{l,1}^{(i)}(t),\ldots ,Y_{l,p_l}^{(i)}(t))^{\top}$, when pointwise evaluations are required.

The population model of a smoothed block is
\begin{equation}\label{4}
Y_{l,j}(t)=\sum_{b=1}^{B_{l,j}}c_{l,j,b} \cdot \varphi_{l,j,b}(t), \ \ \ l \in \{ 1, 2,
\ldots, L\}, \ j \in \{1,2,\ldots, p_l\},
\end{equation}
where the $\varphi_{l,j,1}(t),\ldots,\varphi_{l,j,B_{l,j}}(t)$, $j \in \{1,\ldots, p_l\}$, are known orthonormal basis functions and
$c_{l,j,1},\ldots,c_{l,j,B_{l,j}}$ are unknown random coefficients.
We can rewrite  \eqref{4} as
\begin{equation}\label{5}
\pmb{Y}_l(t)=\pmb{\Phi}_l(t) \cdot \pmb{\mathfrak c}_l,
\end{equation}
where
\begin{equation*}
\pmb{\Phi}_l(t)=
  \begin{pmatrix}
    \pmb{\varphi}^{\top}_{l1}(t) & \pmb{0} & \ldots & \pmb{0} \\
    \pmb{0} & \pmb{\varphi}^{\top}_{l2}(t) & \ldots & \pmb{0} \\
    \ldots & \ldots & \ldots & \ldots \\
    \pmb{0}& \pmb{0} & \ldots & \pmb{\varphi}^{\top}_{lp_l}(t) \\
  \end{pmatrix},
\end{equation*}
$\pmb{\varphi}_{lj}(t)=(\varphi_{l,j,1}(t),\ldots,
\varphi_{{l,j,B_{l,j}}}(t))^{\top}$,
$\pmb{\mathfrak{c}}_l = (c_{l,1,1},\ldots,c_{l,1,B_{l,1}},
\ldots,c_{l,p_l,1},\ldots,c_{l,p_l,B_{l,p_l}})^{\top}$, $t\in I$, $j \in \{1,\ldots,p_l\}$, $l\in \{ 1,2,\ldots, L\}$.

For the sample case, we replace $\mathfrak c_l$ in \eqref{5} by the least-squares estimators $\widehat{\mathfrak c}_l[i]$ such that each block $\pmb A_l[i]$ is encoded by an estimate $\widehat{\mathfrak c}_l[i]$.

In the population case, functional canonical variables are defined as
\[
U_l = \langle \pmb u_l,\pmb Y_l \rangle
=\int_I\pmb u_l^\top(t)\pmb Y_l(t)dt, \ \ \ l \in \{ 1,\ldots, L \},
\]
where $\pmb u_l$ is the vector weight function that lives in the same function space as $\pmb Y_l$, i.e.\ it has a representation
\begin{align}\label{6}
	\pmb u_l(t)=\pmb\Phi_l(t) \cdot \pmb w_l.
\end{align}
In the competitiveness
index example, $l$ represents a feature, e.g., the institutions, so $U_l$
is a weighted average of the indicators in that feature, averaged over time.
The value of $U_l$ depends on the country.

One of the objectives of our analysis is to obtain the mappings
\begin{equation}\label{e:u-est}
    U^{(l)}(\cdot) = \langle {\pmb u}_l^*,\cdot \rangle,
\end{equation}
and to evaluate them for each datum in the sample and attempt to cluster subjects by their values. This is typically done by drawing a scatterplot on
the $(U^{(i)}, U^{(j)})$ plane. Subjects with similar similar $(U^{(i)}, U^{(j)})$ values
are considered as sharing important features.

As $\pmb Y_l$ and $\pmb u_l$ have representations as in \eqref{5} and \eqref{6}, we can write
\begin{equation}\label{e:U}
    U_l = \langle \pmb u_l, \pmb Y_l\rangle
= \pmb w_l^\top \pmb{\mathfrak{c}}_l, \quad l\in \{1,\ldots,L\}.
\end{equation}
Thus, the form of the functional canonical variable corresponding to the random
process $\pmb Y_l$ is determined by the vectors $\pmb{\mathfrak{c}}_l$
and $\pmb w_l$.

In the population case, multiple functional canonical correlation analysis can be presented as the following optimization problem:
\begin{align} \label{e:P3}
\rho = \max_{\pmb u_1,\ldots,\pmb u_L}\sum_{i=1}^L
\sum_{j \in \{1,\ldots,L\}\setminus \{i\}}\Cov(U_i, U_j) \ \
	\text{subject to  } \ \  \sum_{i=1}^L \Var(U_i) = L,
\end{align}
where \[(\pmb u_1^*,\ldots,\pmb u_L^*) = \arg\max_{\pmb u_1,\ldots,\pmb u_L}\sum_{i=1}^L
\sum_{j \in \{1,\ldots,L\}\setminus \{i\}}\Cov(U_i, U_j) \ \
	\text{subject to  } \ \  \sum_{i=1}^L \Var(U_i) = L, \]
By the definition of the $U_i$ in \eqref{e:U}, it is easy to see that this problem is equivalent to
\begin{align*}
\rho = \max_{\pmb u_1,\ldots,\pmb u_L}\sum_{i=1}^L
\sum_{j \in \{1,\ldots,L\}\setminus \{i\}}\Cov(\langle \pmb u_i, \pmb Y_i \rangle, \langle \pmb u_j, \pmb Y_j \rangle) \ \
	\text{subject to  } \ \  \sum_{i=1}^L \Var(\langle \pmb u_i, \pmb Y_i \rangle) = L,
\end{align*}
and thus can be written as in \eqref{e:optimization_problem} by using covariance operators the $\pmb C_{i,j}$, $i,j \in \{1,\ldots,L\}$.
We call the coefficient $\rho$  the canonical functional correlation coefficient,
but we emphasize that our definition is valid for $L>2$.
Previous work, e.g.\ \cite{Leurgans1993}, studied only the case of $L=2$. In the competitiveness
example, $L>2$ allows us to study more than two features of competitiveness.

Similar to the multiple kernel CCA in Section~\ref{sec:kcca}, in the sample case, we replace $\Cov(\langle \pmb u_i, \pmb Y_i \rangle,$ $\langle \pmb u_j, \pmb Y_j \rangle)$ by an estimator $\langle \pmb u_i, \pmb{\widehat{C}}_{i,j} \ \pmb u_j \rangle$, and $\Var(\langle \pmb u_i, \pmb Y_i \rangle)$ by a regularized estimator $\langle \pmb u_i, (\pmb{\widehat{C}}_{i,i}+ \epsilon_n \cdot \pmb I) \ \pmb u_i \rangle$. Using the standard sample covariance estimator $\pmb{\widehat{C}}_{i,j}$ and \eqref{e:U}, it is easy to see that

\begin{equation*}
    \langle \pmb u_i, \pmb{\widehat{C}}_{i,j} \ \pmb u_j \rangle = \langle \pmb w_i, \widehat{\pmb{\mathcal C}}_{i,j} \ \pmb w_j \rangle,
\end{equation*}
with \[\widehat{\pmb{\mathcal C}}_{i,j}=\frac{1}{n-1}\sum_{k=1}^n \widehat{\mathfrak c}_i[k] \widehat{\mathfrak c}_j[k]^{\top} - \frac{n}{n-1} \left(\frac 1n \sum_{l=1}^n \widehat{\mathfrak c}_i[k] \right) \cdot \left(\frac 1n \sum_{\ell=1}^n \widehat{\mathfrak c}_j[\ell] \right)^{\top}.\]
Similarly, for the variance, it holds
\begin{equation*}
    \langle \pmb u_i, (\pmb{\widehat{C}}_{i,i} + \epsilon_n \cdot \pmb I) \ \pmb u_i \rangle = \left\langle \pmb w_i, \left(\widehat{\pmb{\mathcal C}}_{i,i}+ \epsilon_n \cdot \pmb I\right) \ \pmb w_i \right\rangle.
\end{equation*}

Thus, the sample version of the multiple CCA problem for functional data is
\begin{align*}
	\widehat\rho = \max_{\pmb w_1,\ldots,\pmb w_L}\sum_{i=1}^L \sum_{j \in \{1,\ldots,L\}\setminus \{i\}}\pmb w_i^\top \widehat{\pmb{\mathcal C}}_{i,j}\pmb w_j \quad
	\text{subject to} \quad \sum_{i=1}^L\pmb w_i^\top(\widehat{\pmb{\mathcal C}}_{i,i}+\epsilon\cdot \pmb I)\pmb w_i = L,
\end{align*}
and the arg max of this optimization problem gives estimators for \eqref{e:u-est}.
As in Section~\ref{sec:kcca}, the above problem
can be reduced to the following generalized eigenvalue problem
whose solution provides the weights $\pmb w_l$:
\begin{multline}\label{eq:geigen3}
\begin{pmatrix}
  \pmb 0 & \widehat{\pmb{\mathcal C}}_{1,2} & \widehat{\pmb{\mathcal C}}_{1,3} & \ldots & \widehat{\pmb{\mathcal C}}_{1,L}\\
  \widehat{\pmb{\mathcal C}}_{2,1} & \pmb 0 & \widehat{\pmb{\mathcal C}}_{2,3} & \ldots & \widehat{\pmb{\mathcal C}}_{2,L}\\
  \ldots & \ldots & \ldots & \ldots & \ldots\\
  \widehat{\pmb{\mathcal C}}_{L,1} & \widehat{\pmb{\mathcal C}}_{L,2} & \widehat{\pmb{\mathcal C}}_{L,3} & \ldots & \pmb 0
\end{pmatrix}
\begin{pmatrix}
\pmb w_1\\
\pmb w_2\\
\ldots\\
\pmb w_L
\end{pmatrix}\\
= \rho
\begin{pmatrix}
  \widehat{\pmb{\mathcal C}}_{1,1} + \varepsilon\pmb I &\pmb 0 & \pmb 0 & \ldots & \pmb 0\\
  \pmb 0 & \widehat{\pmb{\mathcal C}}_{2,2} + \varepsilon\pmb I & \pmb 0 & \ldots & \pmb 0\\
  \ldots & \ldots & \ldots & \ldots & \pmb 0\\
  \pmb 0 & \pmb 0 & \pmb 0 & \ldots & \widehat{\pmb{\mathcal C}}_{L,L} + \varepsilon\pmb I\\
\end{pmatrix}
\begin{pmatrix}
\pmb w_1\\
\pmb w_2\\
\ldots\\
\pmb w_L
\end{pmatrix}.\nonumber
\end{multline}

\begin{Theorem}
    Let the random variables $Y_l[k]$ take values in a separable Hilbert space, have finite second moments, and Assumptions~\ref{asm2}-\ref{asm4} hold. Moreover, the conditions on the eigenvalues and operators imposed in Theorem~\ref{t:main} are satisfied. Then, the results of Theorem~\ref{t:main} directly carry over to multiple functional CCA described in this section.
\end{Theorem}

\section{Illustrative examples}
\label{s:ex}

\subsection{Agriculture data set for Polish voivodeships}
\label{ss:ex1}

We use agricultural data for Polish voivodeships to demonstrate practical aspects of the described methodology.\footnote{\url{http://stat.gov.pl}} Figure~\ref{fig:nuts} presents the administrative division of Poland into macroregions and voivodeships. Macroregions are broadly equivalent to NUTS 1 units.\footnote{The NUTS classification is a territorial standard valid for the statistical division of the European Union member countries.} In comparison, voivodeships are equivalent to NUTS 2 units.

\begin{figure}[!tbh]
	\centering\includegraphics[width = 0.49\textwidth]{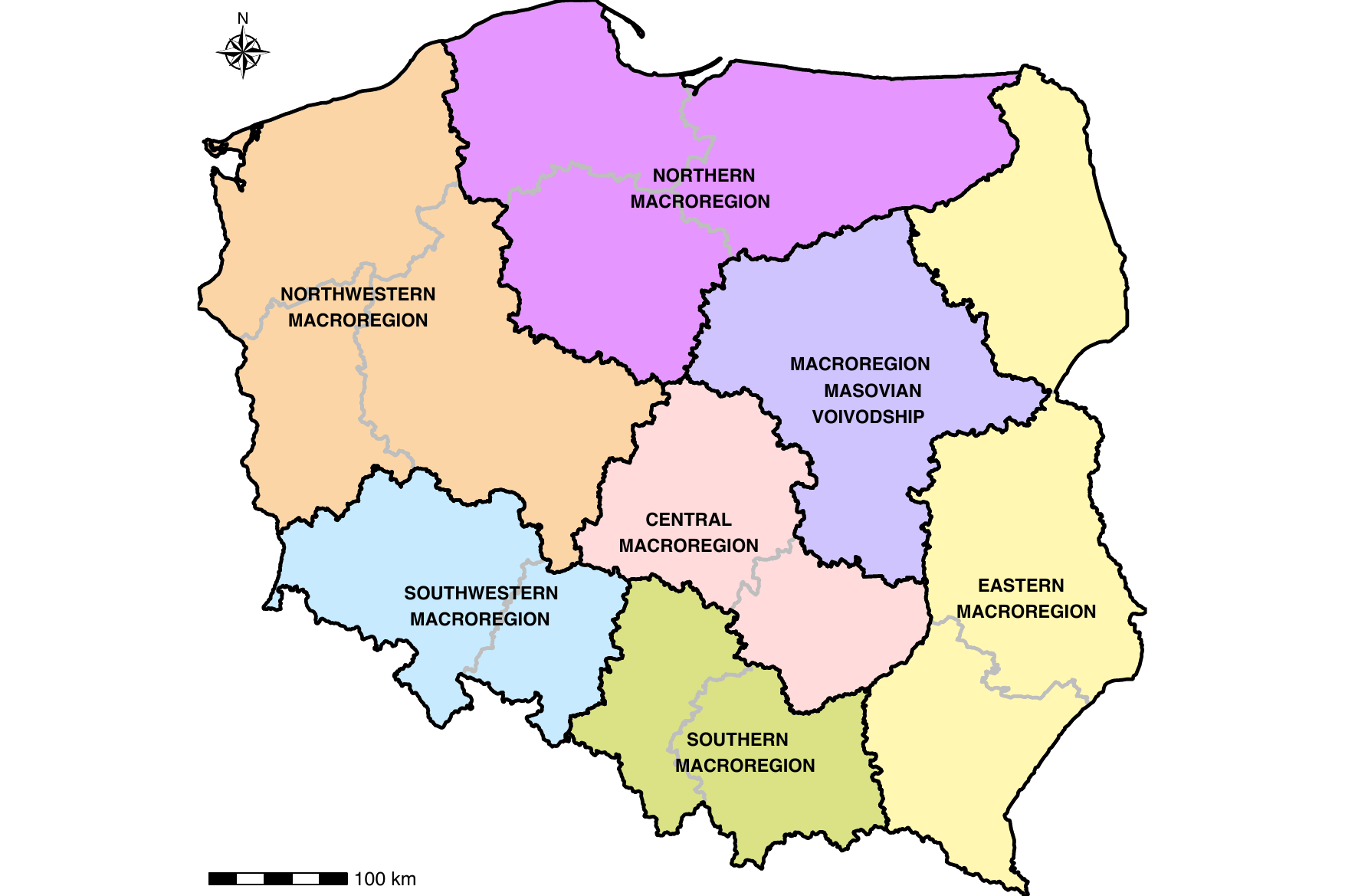}
	\centering\includegraphics[width = 0.49\textwidth]{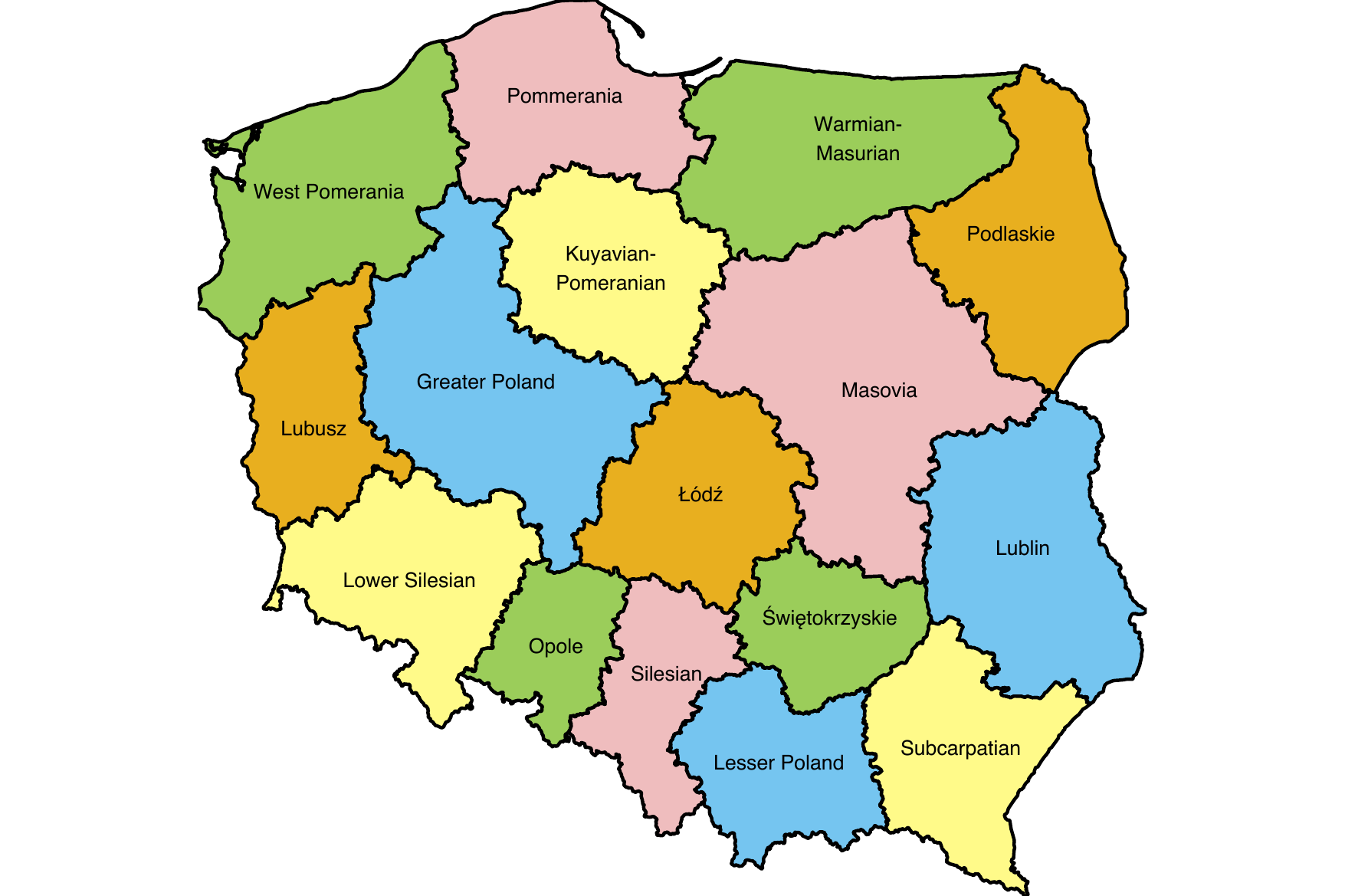}
	\caption{Macroregions (left) and voivodeships (right) in Poland (2018)}
\label{fig:nuts}
\end{figure}

The data consists of yields of thirty crops, expressed in quintals per hectare. These crops were recorded in 2003-2016 ($T = 14$ years) in $n = 16$ Polish voivodeships (administrative units). The analyzed plants are divided into $L=3$ groups:

	\begin{itemize}
		\item Group 1 -- cereals and root crops ($p_1 = 9$ variables): barley, buckwheat, millet, oat, potatoes, rye, sugar beet, triticale and wheat.
		\item Group 2 -- fodders ($p_2 = 6$ variables): clover, field crops, legume fodder, lucerne, root fodder, and serradella.
		\item Group 3 -- fruits and vegetables ($p_3 = 15$ variables): apples, cabbage, carrot, cauliflower, cherries, cucumbers, currants, gooseberry, onion, pears, plums, raspberries, sweet cherries, strawberries, and tomatoes.
	\end{itemize}
	
\subsubsection{Multivariate repeated measures data}
The input data are, therefore, the multivariate repeated measures data $[\pmb A_1[1],\pmb A_2[1],\pmb A_3[1]],\ldots,[\pmb A_1[16],\pmb A_2[16],\pmb A_3[16]]$, with $\pmb A_1[k]\in \mathbb R^{14 \times 9}$, $\pmb A_2[k]\in \mathbb R^{14 \times 6}$, $\pmb A_3[k]\in \mathbb R^{14 \times 15}$.

Based on these data, multiple kernel canonical correlation analysis
was performed (Table~\ref{tab:1}). The multivariate repeated measures
data for 16 voivodeships in the system of the first two multiple kernel canonical variables $(U^{(1)}, U^{(2)})$ are shown in Figure~\ref{fig:cca1}.

\begin{table}[!tbh]
  \caption{Top-3 biggest multiple canonical correlations for the Polish voivodeships data set}
  \begin{tabular}{ccc}
  \toprule
    \multirow{2}{*}{No.} & Multivariate repeated  &Multivariate\\
    &measures data ($\hat\rho$) & functional data ($\hat\rho_F$) \\
    \midrule
    1 & 0.45 & 0.58\\
    2 & 0.29 & 0.37\\
    3 & 0.25 & 0.34\\
    \bottomrule
  \end{tabular}
\label{tab:1}
\end{table}

It is important to note that the reported values are not classical correlation coefficients, but generalized canonical correlations defined by the optimization problems \eqref{e:sample_prob}, and \eqref{e:P3}. For $L=2$, they coincide with the usual canonical correlations and are bounded by one. For $L>2$, however, they quantify the overall strength of association across several sets and are not restricted to the unit interval. The values in Tables~\ref{tab:1} and \ref{tab:3} should therefore be interpreted as generalized correlation coefficients, which are always non–negative but may exceed one without indicating a normalization issue.

Figure~\ref{fig:cca1} shows that the voivodeships belonging to the
same macroregion are located close to one another on the $(U^{(1)}, U^{(2)})$ plot.
This is reasonable because the voivodeships belonging to a given macroregion have
similar temperature, rainfall, and sun exposure, influencing crop yields.
This shows that our algorithm produces beneficial results when some prior knowledge is available. It gives us confidence that its output will be helpful in situations where no prior information can be obtained. The only outlier voivodeship is Opole. This is unsurprising because this voivodeship is considered the best in Poland in terms of agricultural production. Its climate is characterized by hot summers, mild and short winters, early springs, and long, mild autumns. Moreover, 62\% of its area consists of rich brown and clay soils and productive soils in numerous lowland river valleys.  These conditions are ideal for cultivating cereals such as wheat, barley, rapeseed, and sugar beets. The agricultural valorization index\footnote{WWRPP reflects the potential of agricultural production space resulting from natural conditions. It is an integrated indicator that assesses individual habitat elements such as soil quality and suitability, soil water relations, terrain relief, and agroclimate.} (WWRPP) places this region as the best in Poland. The WWRPP for Opole is 81.6 points, compared to the national average of 66.6 points \cite{Stuczynski2007}, see Figure~\ref{fig:wwrpp}. The fact that our algorithm identified this outlier further validates its usefulness.

\begin{figure}[!tbh]
	\centering\includegraphics[width = 1\textwidth]{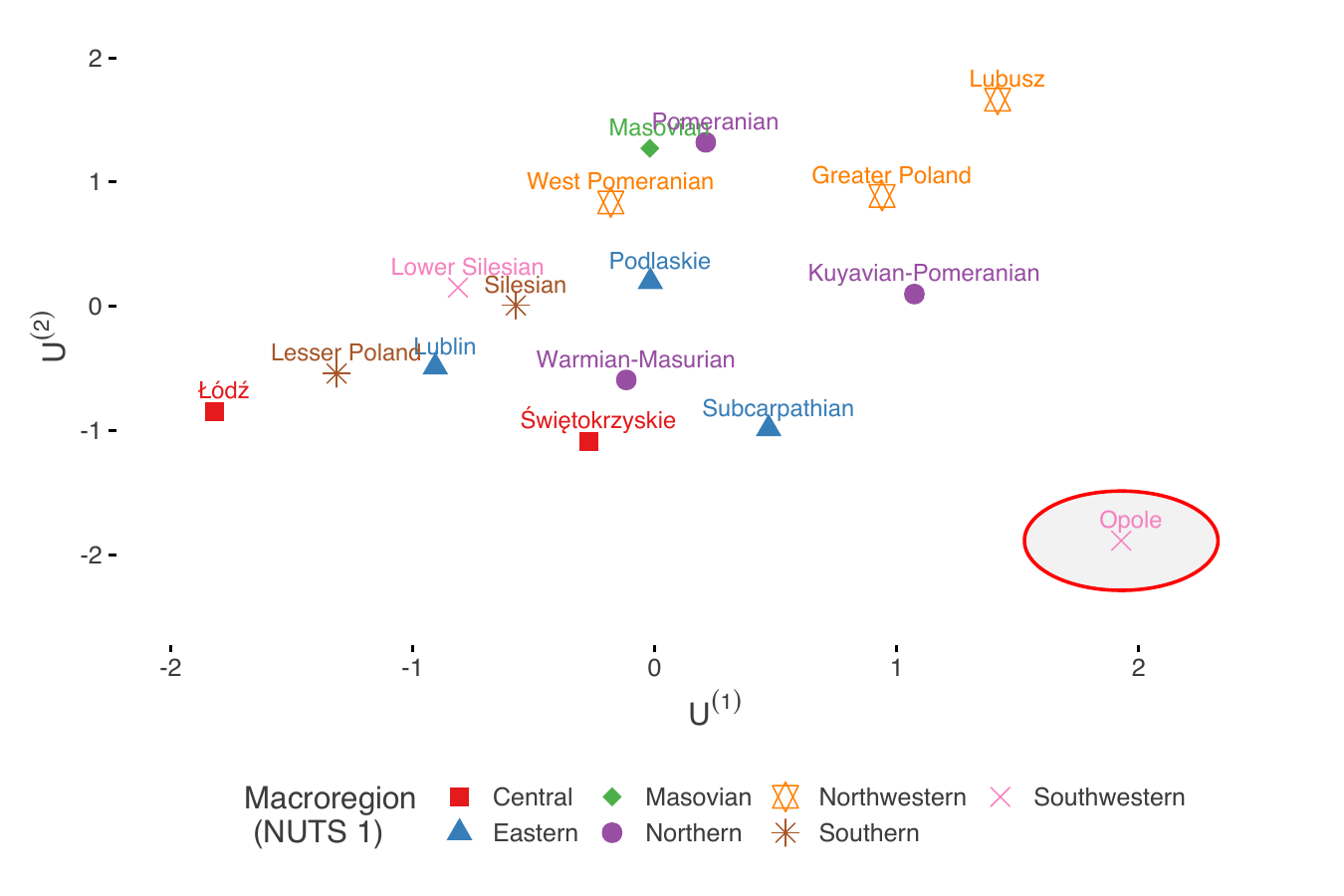}
	\caption{Scatterplot for the optimally transformed feature pairs in the agriculture data set (16 voivodeships and seven regions) in the system of the first two multiple kernel canonical variables $(U^{(1)}, U^{(2)})$. The optimal transformations were determined by multiple kernel CCA for repeated measures data, as described in Section~\ref{sec:kcca}.}
\label{fig:cca1}
\end{figure}

\begin{figure}[!tbh]
	\centering\includegraphics[width = 0.8\textwidth]{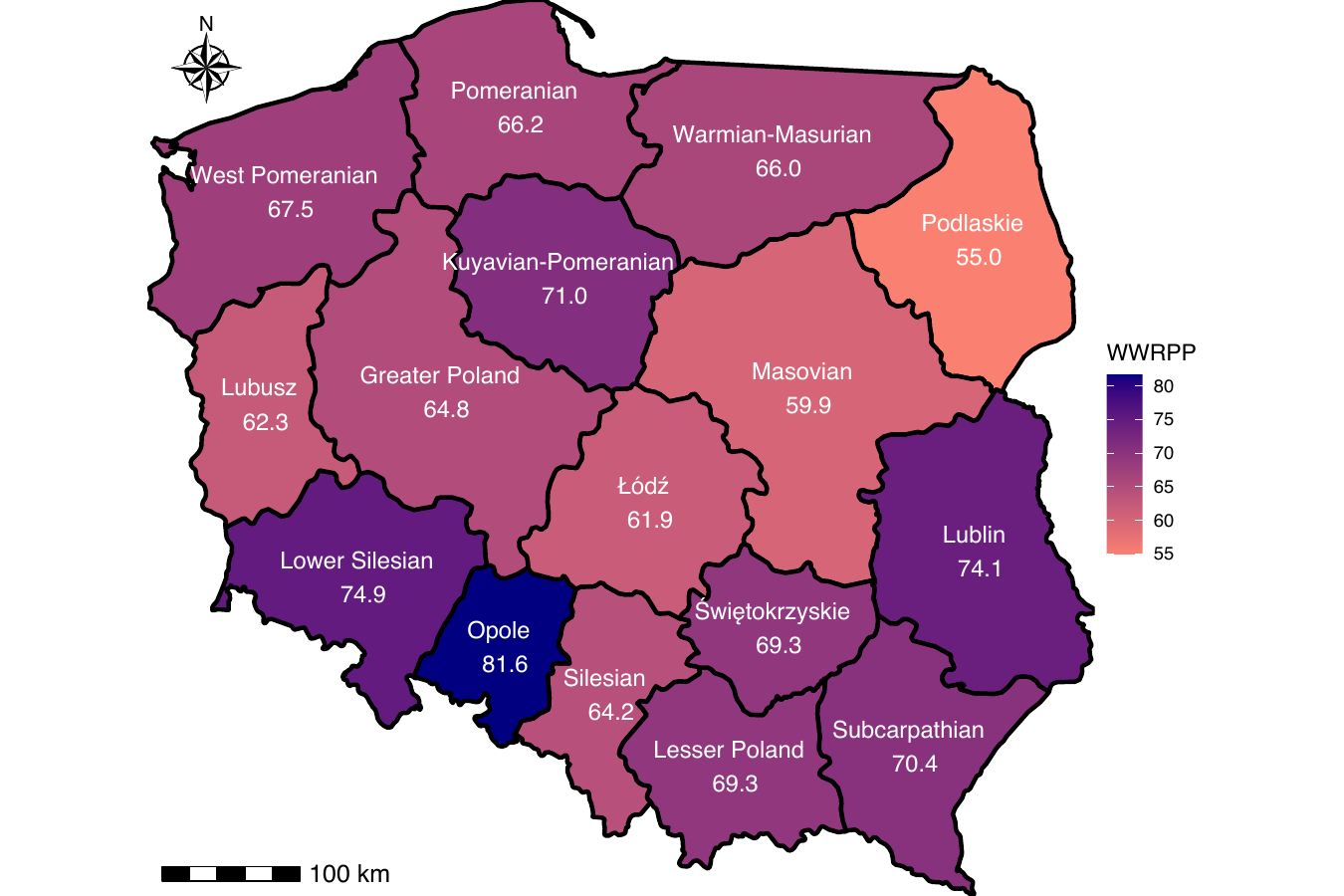}
	\caption{WWRPP index for Poland (country mean = 66.6 points)}
\label{fig:wwrpp}
\end{figure}

\subsubsection{Multivariate functional data}
We used a Fourier basis with nine components ($B_{ij} = 9,\ i \in \{ 1,\ldots, L\}, \ j \in \{ 1,\ldots,p_i\}$) to express the agriculture data as functional data. Using the transformed data, multiple functional canonical correlation analysis was performed (Table~\ref{tab:1}). The multivariate functional data for the 16 voivodeships in the system of the first two multiple functional canonical variables $(U^{(1)}, U^{(2)})$ are shown in Figure~\ref{fig:cca2}, which shows that in the system of multiple canonical variables, the voivodeships are grouped into compact clusters belonging to specific macroregions. Additionally, looking at canonical correlations (Table~\ref{tab:1}), we see that higher values were obtained for the multivariate functional approach. Thus, for these  data, multiple functional canonical variables
have stronger discrimination ability than the multiple kernel canonical variables.

\begin{figure}[!tbh]
	\centering\includegraphics[width = 1\textwidth]{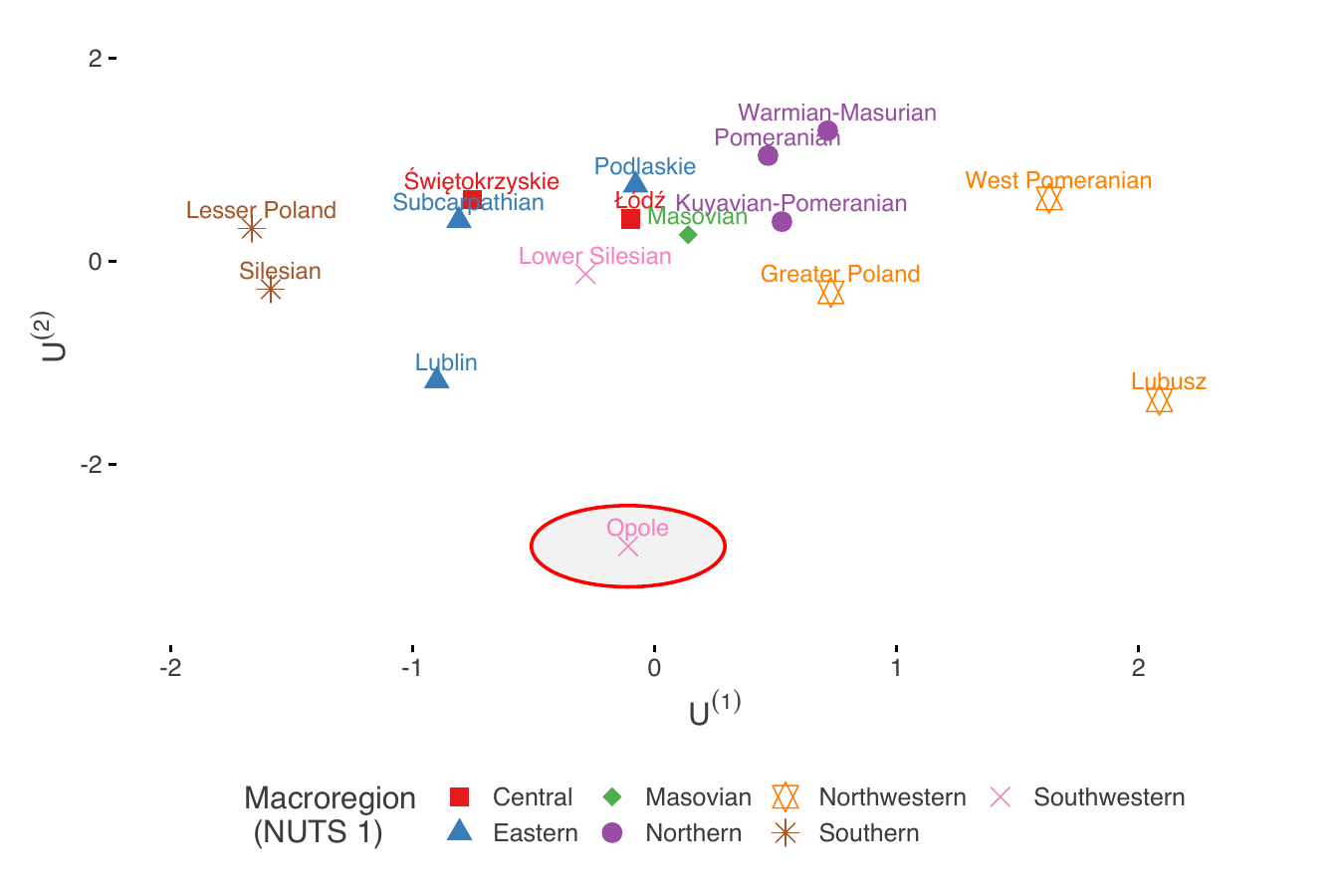}
	\caption{Scatterplot for the optimally transformed feature pairs in the agriculture data set (16 voivodeships and seven regions) in the system of the first two multiple functional canonical variables $(U^{(1)}, U^{(2)})$. The optimal transformations were determined by multiple functional CCA, as described in Section~\ref{sec:fcca}.}
\label{fig:cca2}
\end{figure}

\subsection{Global Competitiveness Index (GCI) data set} \label{ss:gci}
In the second example, we study
the relationships involving $n=115$ countries over $T=10$
years (2008--2017), based on $L=12$ features. For this purpose,
data published by the World Economic Forum (WEF) in its annual
reports\footnote{\url{http://www.weforum.org}} is used.
Established in 1979, the Global Competitiveness Report by the
WEF stands as the most enduring and thorough evaluation of the
factors influencing economic development. These are comprehensive data,
exhaustively describing various socio-economic conditions of countries for
which relevant data are available. Table~\ref{tab:2} describes the features
and many scalar indicators (a detailed description can be found in \cite{Gorecki2016})
in each feature employed in the analysis. WEF experts have divided the countries (115)
into five groups (Figure~\ref{fig:map}). These groups are not used in the analysis, but they illustrate the meaning of the Global Competitiveness Index and are used at the end of this section to validate the clustering implied by our analysis.

\begin{figure}[!tbh]
	\centering\includegraphics[width = 1\textwidth]{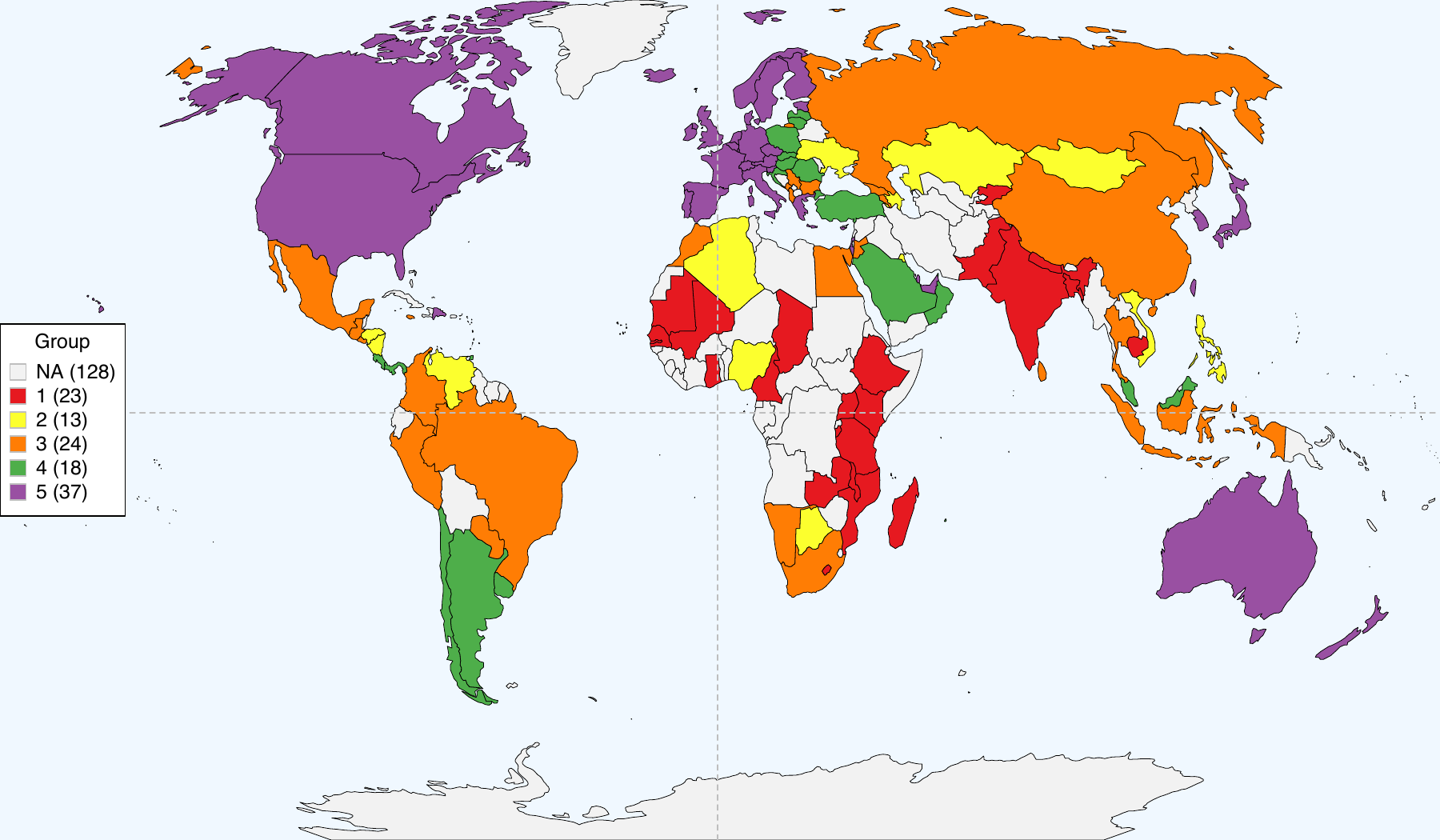}
	\caption{The 115 countries used in our analysis are
highlighted by colors. They are split into five groups by the value of the GCI.
The count of countries in each group is given in  parentheses.
(NA stands for missing data).}
\label{fig:map}
\end{figure}

\begin{table}[!tbh]
    \caption{The 12 features used in the analysis of the GCI data set}
    \begin{tabular}{clc}
    \toprule
    No. ($l$) & Feature & Number of variables, i.e.\ scalar indicators ($p_l$) \\ \midrule
    1. & Institutions & 16 \\
    2. & Infrastructure & 6 \\
    3. & Macroeconomic environment & 3 \\
    4. & Health and primary education & 7 \\
    5. & Higher education and training & 6 \\
    6. & Goods market efficiency & 10 \\
    7. & Labour market efficiency & 6 \\
    8. & Financial market development & 5 \\
    9. & Technological readiness & 4 \\
    10. & Market size & 4 \\
    11. & Business sophistication & 9 \\
    12. & Innovation & 5 \\
    \bottomrule
    \end{tabular}
    \label{tab:2}
\end{table}

We performed the multivariate and the functional  analyses, similarly as in
Section \ref{ss:ex1}. This time, we have $L = 12$ groups, $n = 115$ countries,
and $T = 10$ years. In the case of the MFCCA,
observations are converted to functions by utilizing the Fourier basis
with five basis functions ($B_{lj} = 5,\ l \in \{ 1,\ldots,L\}, \ j \in \{ 1,\ldots,p_l\}$).
Table~\ref{tab:3} displays the results of the two analyses.
It shows that, as in the previous example, the correlations are
slightly higher for the functional approach. The corresponding projections
on the plane can be found in Figures \ref{fig:cca3} and \ref{fig:cca4}.

\begin{figure}[!tbh]
	\centering\includegraphics[width = 1\textwidth]{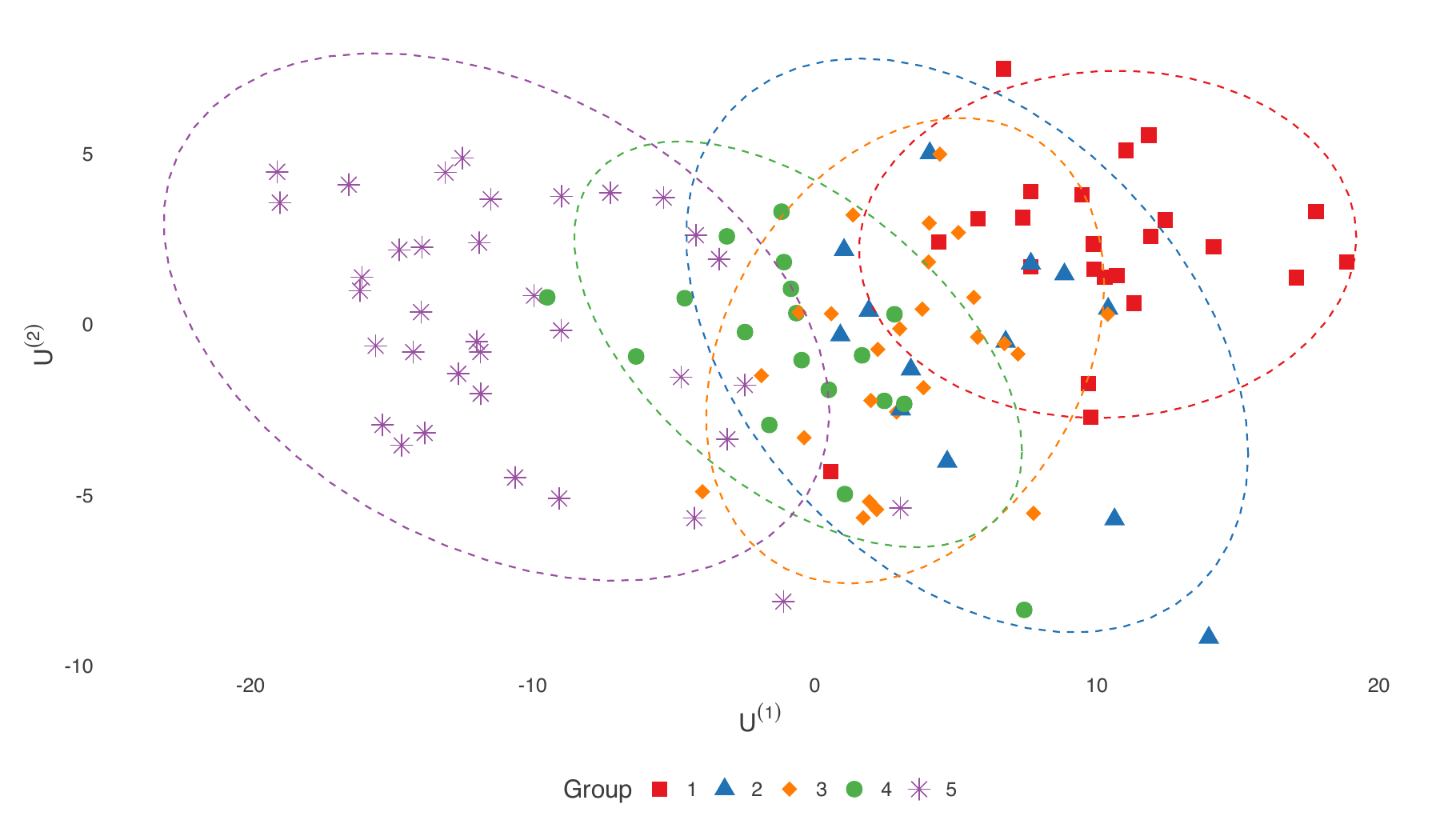}
	\caption{Scatterplot for the optimally transformed feature pairs in the GCI data set (115 countries and five groups) in the system of the first two multiple functional canonical variables $(U^{(1)}, U^{(2)})$ (with 95\% confidence normal ellipses). The optimal transformations were determined by multiple functional CCA, as described in Section~\ref{sec:fcca}.}
\label{fig:cca4}
\end{figure}

\begin{table}[!tbh]
  \caption{Top-3 biggest multiple canonical correlations for GCI data set}
  \begin{tabular}{ccc}
  \toprule
    \multirow{2}{*}{No.} & Multivariate repeated  &Multivariate\\
    &measures data ($\hat\rho$) & functional data ($\hat\rho_F$) \\
    \midrule
    1 & 0.74 & 0.76\\
    2 & 0.28 & 0.29\\
	3 & 0.11 & 0.11\\
    \bottomrule
  \end{tabular}
\label{tab:3}
\end{table}

How is the representation quality assessed, and how is the whole analysis? For this purpose, we decided to use the information that the data is divided into $5$ groups. To evaluate clusterability, we used the Hopkins statistic \cite{hopkins54, lawson90}. This statistic is employed to assess the clustering tendency of a data set. Let $X$ be the data set of $n$ points in the $ d$-dimensional space. We want to test the pair of
hypotheses:

\begin{description}
	\item[H0] The data set $X$ is uniformly distributed (no clusters).
	\item[H1] The data set $X$ is not uniformly distributed (clusters).
\end{description}

Denote by $C_X$ the smallest convex hull that contains $X$.
The Hopkins statistic is calculated with the following algorithm:
\begin{enumerate}
	\item Sample randomly one observation from the data set $X$ and set the counter to $i=1$.  Denote by $w_i$  the Euclidean distance from this observation to the nearest-neighbor observation in $X$.
	\item Generate one new point uniformly distributed in $C_X$
and denote by $u_i$ the Euclidean distance from this point to the nearest-neighbor observation in $X$.
	\item Repeat steps (1) and (2) $m\ll n$ times (typically $m \approx 0.1\cdot n$).
	\item Compute the Hopkins statistic
$$H = \frac{\sum\limits_{i=1}^m u_i^d}{\sum\limits_{i=1}^m (u_i^d + w_i^d)}.$$
\end{enumerate}
Under the null hypothesis, this statistic follows the  $\text{Beta}(m, m)$ distribution. If the data has little structure, the average distance between real points will be similar to that from a uniformly distributed random point to a real point, resulting in a Hopkins statistic value of approximately 0.5. Conversely, if the data are tightly clustered, the distances $w_i$ will be much smaller than those $u_i$, leading to a Hopkins statistic value close to 1.0. The interpretation of $H$ can be understood through the following guidelines \cite{lawson90, wright22}:
\begin{itemize}
\item Low values of $H$ suggest that the observations in $X$ are repelling each other.
\item Values of $H$ close to 0.5 indicate that the observations in $X$ are spatially random.
\item High values of $H$ suggest possible clustering of the observations in $X$.
\item Values of $H$ greater than 0.75 indicate a clustering tendency at the 90\% confidence level.
\end{itemize}
\begin{figure}[!h]
	\centering\includegraphics[width = 1\textwidth]{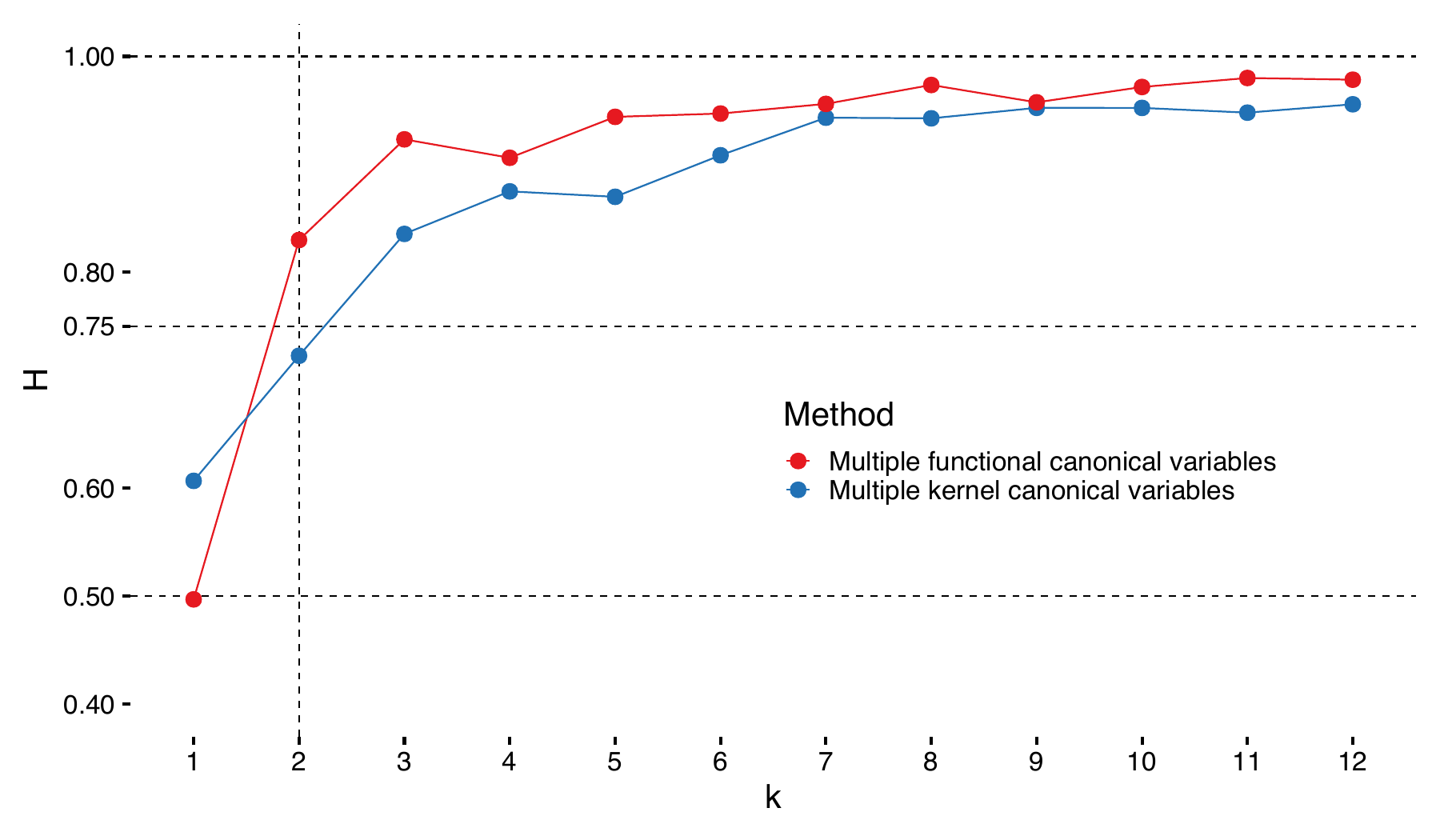}
	\caption{The Hopkins statistic for the GCI data ($m = \lceil115/10\rceil = 12$) vs the  number of components. The horizontal dashed lines indicate the
standard thresholds: $0.5$, $0.75$, and $1$. The
 vertical line at $2$ indicates the projection onto the plane.}
\label{fig:cluster}
\end{figure}
Results are presented in Figure~\ref{fig:cluster}. Due to sampling variability, it is standard to calculate $H$ multiple times and take the average. On the plot, we present for each $k \in \{ 1, \ldots, L\}$, with $L = 12$ (number of multiple canonical correlations) the average value of $H$ for 100 replications.
Multiple functional canonical variables are superior to multiple kernel canonical ones. This method exhibits better clusterability in just two components, indicating that the group structure has been more accurately mapped. For both techniques, the Hopkins statistic approaches one as the number of components increases. However, using only one component ($k=1$) is not sufficient.

\subsection{Calculations details}
Calculations were performed in \texttt{R} 4.2.2 \cite{RCoreTeam2019}, using the packages \textit{fda} \cite{Ramsay2005b}, \textit{geigen} \cite{Hasselman2019}, \textit{hopkins} \cite{wright23} and \textit{RGCCA} \cite{Tenenhaus2017a}.

\section{Concluding remarks} \label{sec:final}
Hotelling's classical canonical correlation analysis has been generalized in two ways in this paper. First, random vectors are replaced by random matrices containing observations of multiple variables at $T$ time points on the same experimental unit (in this case, we are dealing with multivariate repeated measures data, also known as doubly multivariate data).  Relationships between a finite number of random matrices, not limited to two, are considered. In this case, multiple kernel canonical variables are constructed. Secondly, we replace random vectors with multidimensional random processes. In this case, the experimental units are characterized by smooth functions over a time interval (elements of Hilbert space). The experimental units are characterized by multivariate functional data.
Moreover, as before, we consider the relationships between a finite number of multidimensional random processes, not limited to two. In this case, multiple functional canonical variables were constructed. In the case of actual data on yields per hectare for three different groups of crops, multiple functional canonical variables proved superior. A similar result was obtained for the GCI data set.

The proposed methods apply when we have multiple sets of features and aim to examine their dependencies. Additionally, a third dimension is incorporated, such as time or space. The objectives of these techniques are twofold. On the one hand, we aim to present data on a plane to assess dependencies visually.
On the other hand, obtaining vector representations for observations considering time/space allows us to use them in other analyses where direct usage may not be feasible. Canonical correlations for functional data can be applied to such data, but until now, they have only been used for two sets of features. Similarly, in such situations, an appropriate technique of kernel canonical correlations can be employed, which, we hope, captures nonlinear dependencies more effectively. Choosing a specific method is challenging, and generally, it can only be made with at least a preliminary data set analysis. If we have additional information about the groups to which observations belong, we can apply a clusterability assessment, as done in our work.

We have justified our methodology under assumptions weaker than those
used in previous research, even in simpler settings. Specifically, we showed that  certain covariance operators need not be compact and the observational units can be dependent. A high level assumption is formulated that covers both cases.

The proposed operator-based framework for multiple CCA not only unifies 
existing approaches but also offers a principled way to handle dependence 
and non-compactness in high-dimensional or functional data. Future work 
may focus on extending these ideas toward regularized and sparse formulations, 
which could further improve stability and interpretability in complex empirical 
settings.

\section*{Data availability statement}
The data that support the findings of this study are available at \url{http://stat.gov.pl} and \url{http://www.weforum.org}.

\end{document}